\documentclass[draftcls,onecolumn,12pt]{IEEEtran}

\usepackage[latin1]{inputenc}
\usepackage[english]{babel}
\usepackage{graphicx}
\usepackage{float,latexsym}
\usepackage{amsmath, amsfonts, amssymb, latexsym}
\usepackage{cite}

\def\RR{\mathbb{R}}

\def\EE{\mathbb{E}}
\def\R2{\mathcal{R}}
\newcommand{\abs}[1]{\left\vert#1\right\vert}

\DeclareMathOperator{\diag}{diag}
\DeclareMathOperator{\sech}{sech}

\newcommand{\A}{\mathcal{A}}

\newcommand{\N}{\mathcal{N}}

\newcommand{\F}{\mathcal{F}}

\newcommand{\Xc}{\mathcal{X}}

\newcommand{\y}{{\boldsymbol y}}

\newcommand{\X}{{\boldsymbol X}}

\renewcommand{\R}{{\boldsymbol R}}
\renewcommand{\S}{{\boldsymbol S}}

\renewcommand{\A}{{\boldsymbol A}}

\newcommand{\z}{{\boldsymbol z}}

\newcommand{\x}{{\boldsymbol x}}

\renewcommand{\b}{{\boldsymbol b}}
\renewcommand{\b}{{\boldsymbol b}}

\newtheorem{thm}{Theorem}[section]
\newtheorem{cor}[thm]{Corollary}
\newtheorem{lem}[thm]{Lemma}
\newtheorem{claim}[thm]{Claim}

\newtheorem{definition}[thm]{Definition}

\begin{document}

\title{Large-System Analysis of Multiuser Detection with an Unknown Number of Users: A High-SNR Approach}

\author{Adri\`a Tauste Campo, Albert Guill\'en i F\`abregas
\thanks{A. Tauste Campo and A. Guill\'en i F\`abregas
are with the Department of Engineering, University of Cambridge,
Trumpington Street, Cambridge CB2 1PZ, UK, e-mail: {\tt
adria.tauste@eng.cam.ac.uk, albert.guillen@eng.cam.ac.uk}}and Ezio
Biglieri
\thanks{Ezio Biglieri
is with Universitat Pompeu Fabra, c/ Roc Boronat 138, E-08018
Barcelona, Spain, e-mail: {\tt ezio.biglieri@upf.edu}. His work was
supported by the Spanish Ministery of Education and Science under
Project CONSOLIDER-INGENIO 2010 CSD2008-00010 COMONSENS.} }

\maketitle

\begin{abstract}
We analyze multiuser detection under the assumption that the
number of users accessing the channel is unknown by the receiver.
In this environment, users' activity must be estimated along with
any other parameters such as data, power, and location. Our main
goal is to determine the performance loss caused by the need for
estimating the identities of active users, which are not known a
priori. To prevent a loss of optimality, we assume that identities
and data are estimated jointly, rather than in two separate steps.
We examine the performance of multiuser detectors when the number
of potential users is large. Statistical-physics methodologies are
used to determine the macroscopic performance of the
detector in terms of its multiuser efficiency. Special attention is paid to the fixed-point
equation whose solution yields the multiuser efficiency of the
optimal (maximum a posteriori) detector in the large
signal-to-noise ratio regime. Our analysis yields closed-form
approximate bounds to the minimum mean-squared error in this
regime. These illustrate the set of solutions of the fixed-point
equation, and their relationship with the maximum system load.
Next, we study the maximum load that the detector can support for
a given quality of service (specified by error probability).
\end{abstract}

\newpage
\section{Introduction}
In multiple-access communication, the evolution of user activity
may play an important role. From one time instant to the next,
some new users may become active and some existing users inactive,
while parameters of the persisting users, such as power or
location, may vary. Now, most of the available multiuser detection
(MUD) theory is based on the assumption that the number of active
users is constant, known at the receiver, and equal to the maximum
number of users entitled to access the system~\cite{Ver98}. If
this assumption does not hold, the receiver may exhibit a
serious performance loss~\cite{Half96,Poor98}. In~\cite{Big1}, the
more realistic scenario in which the number of active users is
unknown a priori, and varies with time with known statistics, is
the basis of a new approach to detector design. This work
presents a 
large-system analysis of this new type of detectors for Code
Division Multiple Access (CDMA).

Our main goal is to determine the performance loss caused by the
need for estimating the identities of active users, which are not
known a priori. In this paper we restrict our analysis to a
worst-case scenario, where detection cannot improve the performance
from past experience due to a degeneration of the activity model
(for instance, assuming a Markovian evolution of the number of
active users \cite{Osk02,Chen01}) into an independent process
\cite{Tau082}. The same analysis applies to systems where the input
symbols accounting for data and activity are interleaved before
detection. To prevent a loss of optimality, we assume that
identities and data are estimated jointly, rather than in two
separate steps. Our interest is in randomly spread CDMA system in
terms of multiuser efficiency, whose natural dimensions
(number of users $K$, and spreading gain $N$) tend to infinity,
while their ratio (the ``system load'') is kept fixed. In
particular, we consider the optimal maximum a posteriori (MAP)
multiuser detector, and use tools recently adopted from statistical
physics~\cite{Tan02,Guo05,Mont08,Nis92,Tan08, Guo09}. Of special relevance
in our analysis is the decoupling principle introduced
in~\cite{Guo05} for randomly spread CDMA. The general results
derived from  asymptotic analysis are validated by simulations run
for a limited number of users~\cite{Tan08}.

The results of this paper focus on the degradation of multiuser
efficiency when the uncertainty on the activity of the users grows
and the SNR is sufficiently large. We go one step beyond the
application of the large-system decoupling principle \cite{Tan02,
Guo05} and provide a new high-SNR analysis on the space of
fixed-point solutions showing explicitly its interplay with the
system load for a non-uniform ternary and parameter-dependent input
distribution. By expanding the minimum mean square error for large
SNR, we obtain tight closed-form bounds that describe the large CDMA
system as a function of the SNR, the activity factor and the system
load. In addition, some trade-off  results between these quantities are
derived. Of special novelty here is the study of the impact of the
activity factor in the CDMA performance measures (minimum
mean-square error, and multiuser efficiency). In particular, we provide necessary and sufficient conditions
on the existence of single or multiple fixed-point solutions
as a function of the system load and SNR. 
Finally, we analytically identify
the region of  ``meaningful" multiuser efficiency solutions with their associated maximum system loads, and
derive consequences for engineering problems of practical interest.

This paper is organized as follows. Section~\ref{section:model}
introduces the system model and the main notations used
throughout. Section~\ref{section:main} derives the large-system
central fixed-point equation, and analytical bounds to the MMSE.
Based on these results, Section~\ref{section:main2} discusses the
interplay of maximum system load and multiuser efficiency. 
Finally, Section~\ref{section:conclusions} draws some
concluding remarks.

\section{System Model}
\label{section:model}

We consider a CDMA system with an unknown number of
users~\cite{Big1}, and examine the optimum user-and-data detector.
In particular, we study randomly spread direct-sequence (DS) CDMA
with a maximum of $K$ active users:
\begin{equation}\label{eq:1.1}
    \y_t=\S \A \,\b_t + \z_t,
\end{equation}
where $\y_t\in\RR^{N}$ is the received signal at
time $t$, $N$ is the length of the spreading
sequences, $\S\in \RR^{N\times K}$ is the matrix of the sequences,
$\A=\diag(a_1,\dotsc,a_K)\in \RR^{K\times K}$ is the diagonal matrix
of the users' signal amplitudes, $\b_t=(b_t^1,\dots,
b_t^K)\in\RR^{K}$ is the users' data vector, and $\z_t$ is an
additive white Gaussian  noise vector with i.i.d.\ entries $\sim
\N(0,1)$. We define the system's activity
rate as $\alpha \triangleq \Pr\{\text{user $k$ is active}\}$, $1\leq
k\leq K$. Active users employ binary phase-shift keying (BPSK) with
equal probabilities. This scheme is equivalent to one where each
user transmits a ternary constellation $\Xc \triangleq\{-1,0,+1\}$
with probabilities $\Pr\{b_t^k = -1\}=\Pr\{b_t^k =
+1\}=\tfrac{\alpha}{2}$ and $\Pr\{b_t^k = 0\}=1-\alpha$. We define
the  maximum system load as $\beta \triangleq \frac{K}{N}$.

In a static channel model, the detector operation remains invariant
along a data frame, indexed by $t$, but we often omit this time
index for the sake of simplicity. Assuming that the receiver knows
$\S$ and $\A$, the a posteriori probability (APP) of the transmitted
data has the form
\begin{equation}\label{eq:1.2}
p(\b|\y,
\S,\A)=\frac{1}{\sqrt{2\pi}}e^{-\frac{\|\y-\S\A\,\b\|^2}{2}}\,\frac{p(\b)}{p(\y|\S,\A)}.
\end{equation}
Hence, the maximum a posteriori (MAP) joint activity-and-data
multiuser detector solves
\begin{equation}\label{eq:1.3}
 \hat \b= \arg \max_{\b\in \Xc^K} p(\b|\y,\S, \A).
\end{equation}
Similarly, optimum detection of single-user data and activity is
obtained by marginalizing over the undesired users as follows:
\begin{equation}\label{eq:1.4}
 \hat b^k= \arg \max_{b^k} \sum_{\b\setminus b^k}p(\b|\y,\S,\A).
\end{equation}

\subsection{The decoupling principle}
In a communication scheme such as the one modeled
by~\eqref{eq:1.2}, the goal of the multiuser detector is to infer
the information-bearing symbols given the received signal $\y$ and
the knowledge about the channel state. This leads naturally to the
choice of the partition function $Z(\y,\S)=p(\y \mid \S)$. The
corresponding free energy, normalized by the number of users becomes
\cite{Tan02}
\begin{equation}\label{eq:1.8}
\F_K\triangleq-\frac{1}{K}\ln p(\y \mid \S)
\end{equation}
To calculate this expression we make the self-averaging assumption,
which states that the randomness of~\eqref{eq:1.8} vanishes as $K
\to \infty$. This is tantamount to saying that the free energy per
user $\F_K$ converges in probability to its expected value over the
distribution of the random variables $\y$ and $\S$, denoted by
\begin{equation}\label{eq:1.9}
\F\triangleq\lim_{K\rightarrow \infty} \EE\left\{-\frac{1}{K}\ln
p(\y\mid \S)\right\}.
\end{equation}
Evaluation of~\eqref{eq:1.9} is made possible by the \emph{replica
method} \cite{Nis92,Mont08}, which consists of introducing $n$
independent replicas of the input variables, with corresponding
density $p^n(\y|\S)$, and computing $\F$ as follows:
\begin{equation}\label{eq:1.10}
\F=-\lim_{n\to 0} {\partial\over \partial n}\left(
\lim_{K\to\infty}{1\over K}\ln \EE\{p^n(\y\mid \S)\}\right).
\end{equation}

To compute \eqref{eq:1.10}, one of the cornerstones in large
deviation theorem, the Varadhan's theorem~\cite{El85}, is invoked to
transform the calculation of the limiting free energy into a
simplified optimization problem, whose solution is assumed to
exhibit symmetry among its replicas. More specifically, in the case
of a MAP individually optimum detector, the optimization yields a
fixed-point equation, whose unknown is a single operational
macroscopic parameter, which is claimed to be the multiuser
efficiency \footnote{The multiuser efficiency reflects the degradation factor of SNR due to
interference\cite{Ver98}.} of an equivalent Gaussian channel~\cite{Guo05}. Due to
the structure of the optimization problem, the multiuser efficiency
must minimize the free energy. The above is tantamount to
formulating the \emph{decoupling principle}:

\begin{claim}\label{thm0}
\cite{Guo05,Tan08} Given a multiuser channel, the distribution of
the output $\hat b^k$ of the individually optimum (IO) detector,
conditioned on $b^k=b$ being transmitted with amplitude $a$,
converges to the distribution of the posterior mean estimate of
the single-user Gaussian channel
\begin{equation}\label{eq:1.12}
y=\sqrt{\gamma}b^k+\frac{1}{\sqrt{\eta}}z,
\end{equation}
where $z\sim \N(0,1)$,
and $\eta$, the multiuser efficiency, is the solution of the
following fixed-point equation:
\begin{equation}\label{eq:1.13}
\eta^{-1}=1+\beta\EE_{\gamma}\left[\gamma\textrm{MMSE}\left(\eta\gamma,
\alpha\right)\right].
\end{equation}
If~\eqref{eq:1.13} admits more than one solution, we must choose
the one minimizing the free energy function
\begin{equation}\label{eq:1.14}
\F=-\EE\left[\int p(y\mid b^k)\ln p(y\mid b^k) \textrm{d}y\right]
-\frac{1}{2}\ln\frac{2\pi
e}{\eta}+\frac{1}{2\beta}\left({\eta}\ln\frac{2\pi}{\eta}\right).
\end{equation}
In \eqref{eq:1.13}, \eqref{eq:1.14}, $p(y|b^k)$ is the transition
probability of the large-system equivalent single-user Gaussian
channel described by~\eqref{eq:1.12}, and
\begin{equation}\label{eq:1.15}
\textrm{MMSE}(\eta\gamma,\alpha)\triangleq\EE\left[\left(b^k-\hat
b^k\right)^2\right]
\end{equation}
denotes the minimum mean-square error in estimating $b^k$ in
Gaussian noise with amplitude equal to $\sqrt{\gamma}$, where
$\hat b^k=\EE\left[b^k|y\right]$ is the posterior mean estimate,
known to minimize the MMSE~\cite{Poor88}.
\end{claim}

\subsection{A note on the validity of the replica
method}\label{subsection:val}
The replica method is known to accurately approximate experimental data and is consistent with previous theoretical work ~\cite{Tse99,Tse00,Verd02}. The replica
method analysis relies on four unproved assumptions: i) the
self-averaging property of the free energy, ii) the replica symmetry of
the fixed-point solution, iii) the exchange of order of limits and
iv) the analytic continuation of the replica exponent to real values.
Although the validation of the mathematical rigor of these
assumptions is still an unsolved problem, there has been some recent
progress in this direction~\cite{Tal03,Mont06,Guo07,Kor09}.

\section{Large-system multiuser efficiency}
\label{section:main}
We illustrate here the behavior of multiuser efficiency and system
load in the high-SNR region corresponding to detection with an
unknown number of users. We start by shaping our problem into the
statistical-physics framework~\cite{Tan02,Guo05}. As mentioned
earlier, the multiuser detector metric is regarded as the energy of
a system of particles at state $\X $. Therefore, the 
partition function $Z(\X)=\sum_{\x}\exp\left(-{\varepsilon(\x)/
T}\right)$ corresponds to the output density given the channel
information, i.e.,
$p(\y|\S)=(2\pi)^{-1/2}\sum_{\b}p(\b)\exp\left(-\|\y-\S\A\,\b\|^2/2\right)$.

The energy operator $\varepsilon(.)$, as derived from the free
energy, is related to the logarithm of the joint distribution
$p(\y\mid \b,\S)p(\b)$:
\begin{equation}\label{eq:2.2}
    \varepsilon(\b)= \|\y-\S\A\b\|^2-2\ln p(\b)
\end{equation}

We can now invoke the decoupling principle (Claim \ref{thm0}) in the
multiuser system~\eqref{eq:1.1}, so as to use its single-user
characterization. By doing this, the system's performance can be characterized 
by that of a bank of $K$ 
scalar Gaussian channels~\eqref{eq:1.12}, where $K$ represents the
maximum number of users. The input distribution for an arbitrary
BPSK user $k$ takes values $\Xc=\{-1,0,+1\}$ with probabilities
$\frac{\alpha}{2},1-\alpha$ and $\frac{\alpha}{2}$, respectively,
the signal amplitudes from matrix $\A$ are assumed to be constant, i.e, $a_k=\sqrt{\gamma}$ $\forall k$, where
$\gamma$ is the SNR per active user (referred to as SNR), and the
inverse noise variance is equal to the multiuser efficiency $\eta$.
Hence, $\eta$ is the solution of the fixed-point
equation~\eqref{eq:1.13} that minimizes~\eqref{eq:1.14}, where the
MMSE is given by~\eqref{eq:1.15}. More generally,
the analysis presented in this paper can be easily extended to
$a_k$ coefficients with different statistics, like for example
those induced by Rayleigh fading.

By applying Claim \ref{thm0} \cite{Guo05} which holds under the
assumptions of the replica method, the fixed-point equation of the
user-and-data detector can be stated as follows:

\begin{cor}\label{thm1}
Given a randomly spread DS-CDMA system with constant equal power per
user, the large-system multiuser efficiency of an individually
optimum detector that performs MAP estimation of users' identities
and their data under BPSK transmission is the solution of the
following fixed-point equation
\begin{equation}\label{eq:2.3}
\eta=\frac{1}{1+\beta\left(\gamma\left[\alpha-\int\frac{1}{\sqrt{2\pi}}e^{\frac{-y^2}{2}}\frac{\alpha^2\sinh(\eta
\gamma-y\sqrt{\eta\gamma})}{\alpha\cosh(\eta\gamma-y\sqrt{\eta
\gamma})+(1-\alpha)e^{\eta\frac{\gamma}{2}}}\textrm{d}y\right]\right)}
\end{equation}
that minimizes the free energy \eqref{eq:1.14}.
\end{cor}

\begin{proof}
See Appendix \ref{app:proof_mmse}.
\end{proof}

\begin{figure} [!htbp]
         \begin{center}
         \includegraphics[width=0.9\columnwidth]{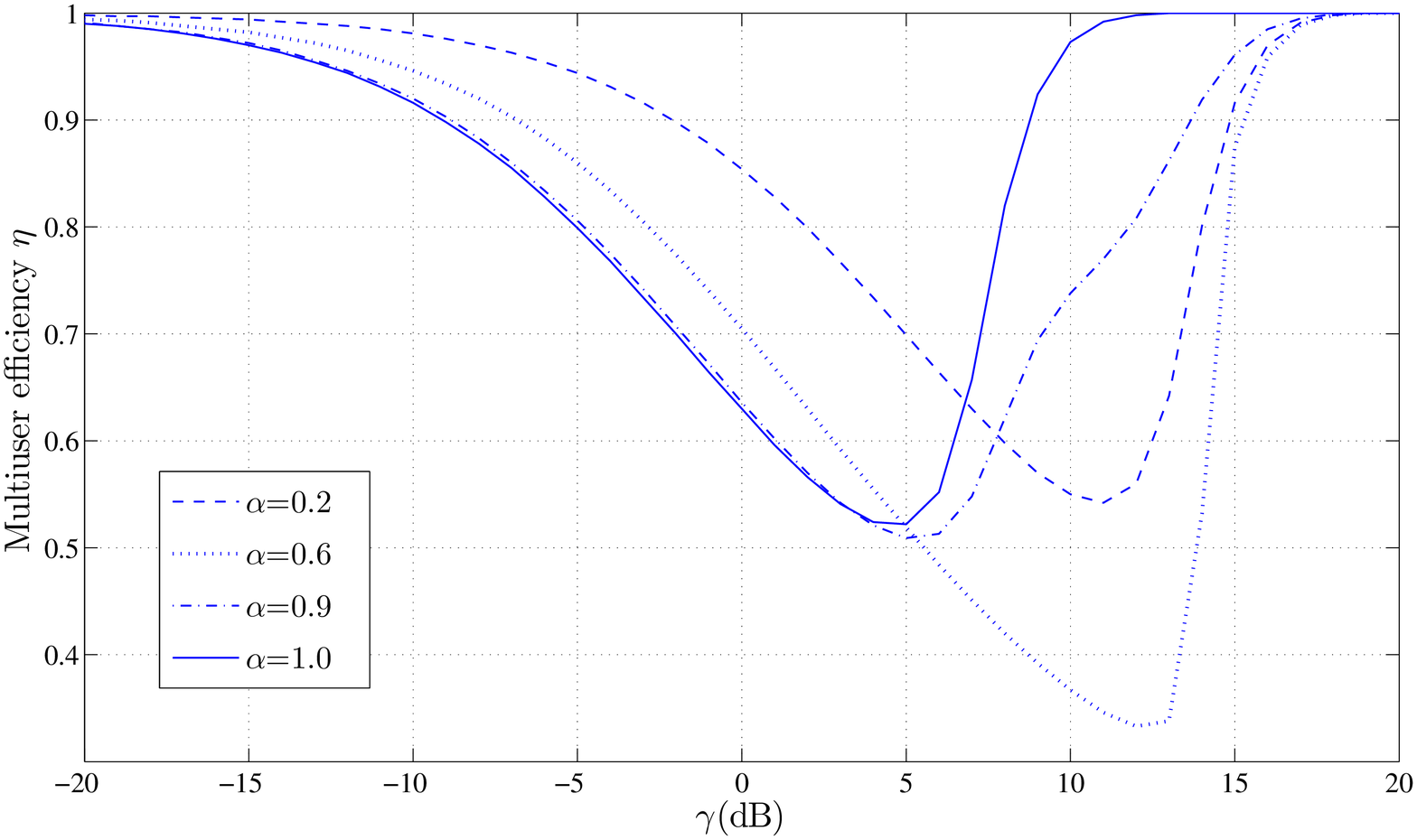}
          \caption{{Large-system multiuser efficiency of the user-and-data detector under MAP with prior knowledge of $\alpha$ and $\beta$=3/7.}}
          \label{fig1}
            \end{center}
                       \vspace{-7mm}

\end{figure}

Our approach differs from that in~\cite{Guo05,Mull03,Tan02}, as the
fixed-point equation~\eqref{eq:2.3} also includes the prior
distribution on the users' activity in a static channel. Under MAP
estimation, detection requires the knowledge not only of the prior
information of the data, but also of the activity rate $\alpha$.
Thus, the fixed-point equations depend on MMSE, SNR, and system
load. Numerical solutions vs.\ SNR at a load $\beta=3/7$ are shown
in Fig.~\ref{fig1}. Plots like this one illustrate how the multiuser
efficiency is affected by the level of noise and interference, and
by the uncertainty in the users' activity rate. For low SNR, noise
dominates, and the performance of the MMSE and the multiuser
efficiency is degraded as $\alpha$ grows, since the presence of more
active users adds more noise to the system. On the other hand, as we
shall discuss later, for high SNR the MMSE strongly depends on the
minimum distance between the transmitted symbols, and the activity
rate here plays a secondary role. Hence, the gap between the
multiuser efficiencies with $\alpha=1$ and $\alpha\ne 1 $ for larger
SNR is due to the fact that the former constellation has twice the
minimum distance of the latter. We can observe clearly the
transition behavior from low to high SNR for values of $\alpha$
approaching $1$. Moreover, when $\alpha=1$, \eqref{eq:2.3} reduces
to the fixed-point equation for the classical assumption, in which all
users are active and transmit a binary antipodal
constellation~\cite{Guo05}:
\begin{equation}\label{eq:2.4}
\eta^{-1}=1+\beta\left(\gamma\left[1-\int\frac{1}{\sqrt{2\pi}}e^{-y^2/2}\tanh(\eta
\gamma-y\sqrt{\eta\gamma})\textrm{d}y\right]\right).
\end{equation}
In this case, it can be shown that, for high SNR, we have
$\textrm{MMSE}(\eta\gamma,\alpha=1)\approx
\sqrt{\frac{2\pi}{\eta\gamma}}e^{-\eta\gamma/2}$. In fact, the
following general result holds:
\begin{lem}\label{lem1}
\cite{Loz06} For large output SNR, the MMSE of a system
transmitting an equiprobable $M$-ary normalized constellation with
minimum Euclidean distance $d$ in a Gaussian channel with noise
variance $1/\eta$ is
\begin{equation}\label{eq:2.5}
\textrm{MMSE}(\eta\gamma,\alpha=1)
=\kappa(\eta\gamma)e^{-d^2\eta\gamma/8}
\end{equation}
with $\kappa_1(\eta\gamma)\leq \kappa(\eta\gamma)\leq \kappa_2$,
where $\kappa_1(\eta\gamma)=\mathcal{O}(1/\sqrt{\eta\gamma})$ and
$\kappa_2$ is a constant, given by the maximum distance between
neighboring symbols.
\end{lem}

For the entire range of activity rates, i.e., $\alpha\in [0,1]$, we
can derive  lower and  upper bounds illustrating analytically the
transition between the classical assumption ($\alpha=1$) and the
cases where the activity is also detected ($\alpha < 1$) for large SNR.
Our calculations bring about a new analytical framework to deal
with large-system analysis, as we will see in the next section. Our
bounds are consistent with Lemma~\ref{lem1} and the lower bound
includes the case $\alpha=1$. The general result is stated as
follows.
\begin{thm}\label{thm2}
The MMSE of joint user identification and data detection in a large
system with an unknown number of users has the following behavior,
valid for sufficiently large values of the product $\eta\gamma$:
\begin{equation}\label{eq:2.6}
2\sqrt{\frac{\alpha(1-\alpha)} {{\pi
\eta\gamma}}} e^{-\eta\gamma/8}\leq\textrm{MMSE}(\eta\gamma,\alpha)\leq
2\alpha e^{-\eta\gamma/2}+
\sqrt{\frac{\pi\alpha(1-\alpha)}
{\eta\gamma}}e^{-\eta\gamma/8}
\end{equation}
\end{thm}
\begin{proof}
See Appendix \ref{app:proof_mmse_limits}.
\end{proof}

Bounds in \eqref{eq:2.6} describe explicitly, in
the high-SNR region, the relationship between the MMSE, the users'
activity rate, and the effective SNR ($\eta\gamma$). In Fig.~\ref{fig2} these bounds are
compared to the true MMSE values  as a function of $\eta\gamma$ for
fixed $\alpha$. It can be seen that the uncertainty about the users'
activity modifies substantially the exponential decay of the MMSE
for high SNR. In fact, a value of $\alpha$ different from $1$
causes the MMSE to decay by $\exp(-\eta\gamma/8)$, rather than by
$\exp(-\eta\gamma/2)$, which would be the case when all users are
active. Furthermore, we can observe that, for sufficiently large
effective SNR, the behavior vs. $\alpha$ of the optimal detector is
symmetric with respect to $\alpha=1/2$, which corresponds to the
maximum uncertainty of the activity rate.
Figure~\ref{fig3} shows that for large values of the product $\eta\gamma$, the MMSE essentially depends on
the minimum distance between the inactivity symbol $\{0\}$ and the
data symbols $\{-1, 1\}$, and thus users' identification prevails
over data detection. Summarizing, the dependence of the MMSE must
be symmetrical with respect to $\alpha=1/2$, since it reflects the
impact of prior knowledge about the user's activity into the
estimation.

\begin{figure} [!htbp]
         \begin{center}
         \includegraphics[width=0.6\columnwidth,, angle=270]{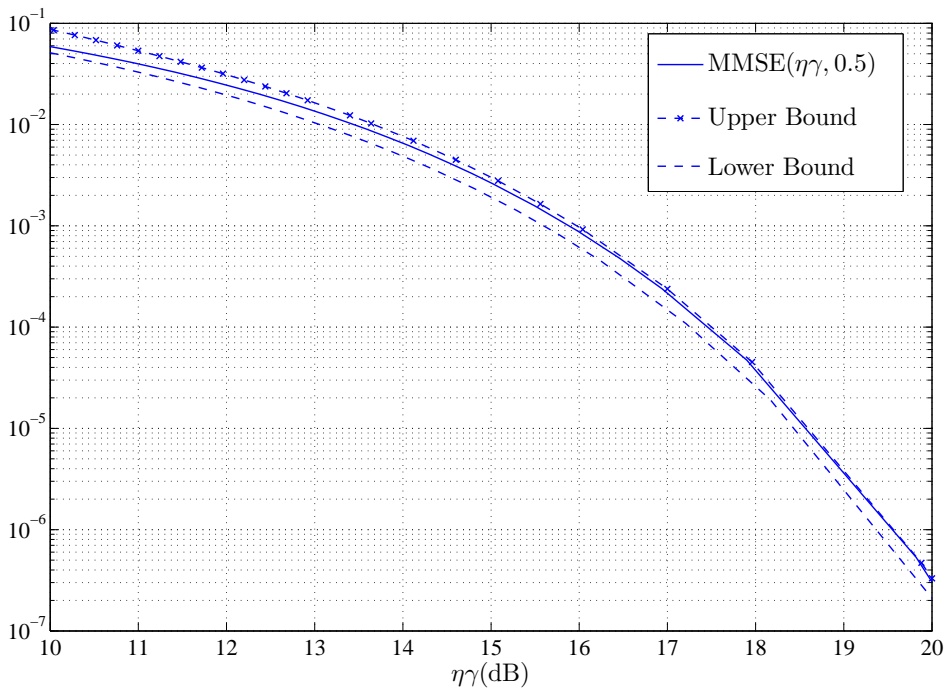}
          \caption{A comparison of the exact MMSE value with its upper
            and lower bounds for $\alpha=0.5$ and $\eta\gamma\in[10,20]$~dB}
            \label{fig2}
            \end{center}
                       \vspace{-7mm}

\end{figure}

\begin{figure} [!htbp]
         \begin{center}
         \includegraphics[width=0.6\columnwidth,, angle=270]{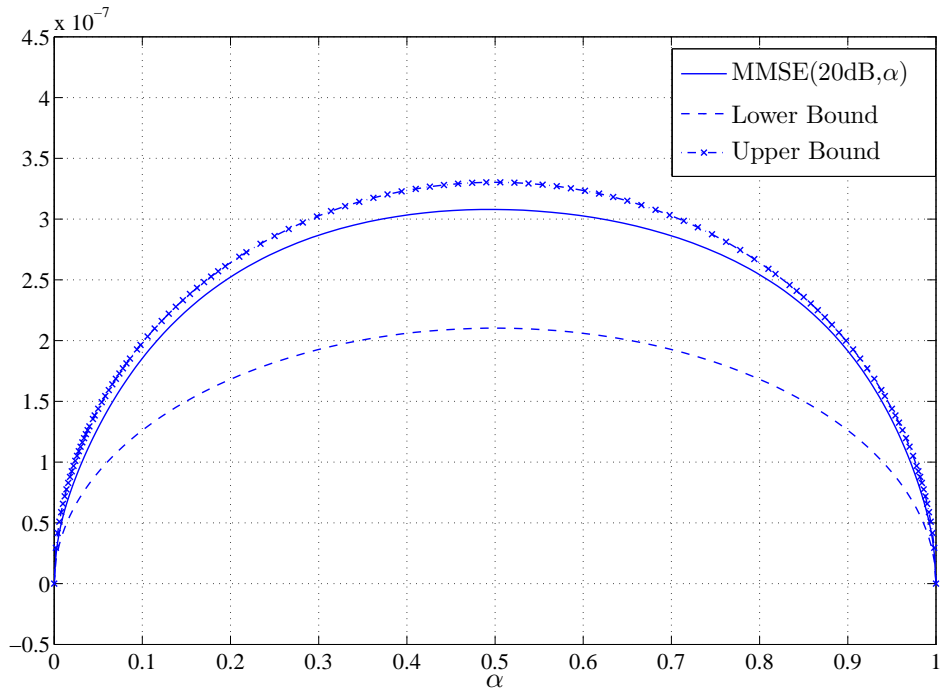}
          \caption{A comparison of the exact MMSE value with its upper
            and lower bounds for $\eta\gamma=20$~dB and $\alpha \in [0,1]$.}
            \label{fig3}
            \end{center}
                       \vspace{-7mm}

\end{figure}

\section{Maximum System Load and Related Considerations}\label{section:main2}

Recall the definition of maximum system load $\beta =
\frac{K}{N}$, where $K$ is the maximum number of users accessing
the multiuser channel. When the number of active users is unknown,
and there is a priori knowledge of the activity rate, the actual
system load is $\beta'=\alpha \beta$. In this section, we focus on
$\beta$ and study some of its properties. Notice that, given an
activity rate, results for the actual system load follow trivially.


\subsection{Solutions to the large-system fixed-point equation}
We characterize the behavior of the maximum system load subject to
quality-of-service constraints. This helps shedding light into
the nature of the solutions of the fixed-point
equation~\eqref{eq:2.3}. In particular, there might be cases
where~\eqref{eq:2.3} has multiple solutions. These solutions
correspond to the solutions appearing in any simple mathematical model of
magnetism based on the evaluation of the free energy with the
fixed-point method~\cite{Nis92}. They represent what in the
statistical physics parlance is called \emph{phase coexistence} (for
example, this occurs in ice  or liquid phase of water at $0^{\circ}$C).
In particular, at low temperatures, the magnetic system might have three
solutions $0\leq\Psi_1<\Psi_2<\Psi_3\leq1$. Solutions $\Psi_1$ and
$\Psi_3$ are stable: one of them is globally stable (it actually
minimizes the free energy), whereas the other is metastable, and a
local minimum. Solution $\Psi_2$ is always unstable, since it is is
a local maximum. The ``true'' solution is therefore given by
$\Psi_1$ and $\Psi_3$, for which the free energy is a minimum. The
same consideration applies also to our multiuser detection problem
where multiuser efficiencies for the IO detector might vary
significantly depending on the value of the system load and SNR.
More specifically, for sufficiently large SNR, stable solutions may switch
between a region that approaches the single-user performance
($\eta=1-\epsilon_1$) and a region approaching the worst performance
($\eta= \epsilon_0$), for $0<\epsilon_1,\epsilon_0\ll1$. Following
previous literature~\cite{Tan02}, we shall call the former solutions
\emph{good} and the latter \emph{bad}. When the solution is unique,
due to low or high system load, the multiuser efficiency is a
globally stable solution that lies in either the good or the bad
solution region. Then, for given system parameters, the set of
\emph{operational} (or globally stable) solutions is formed by
solutions that are part of these sets and minimize the free
energy.

The existence of good and bad
solutions are critical in our problem. From a computational perspective, we are
particularly interested in single solutions, either bad or good,
that surely avoid metastability and instability. These solutions
belong to a specific subregion within the  bad  and  good regions,
and appear for low and high SNR, respectively. From an
information-theoretic perspective, it might seem that the true
solutions should capture all our attention. However, it has been
shown that metastable solutions appear in suboptimal
belief-propagation-based multiuser detectors, where the system is
easily attracted into the bad solutions region (corresponding
to low multiuser efficiency), due to initial configurations that are  far
from the true solution~\cite{Tan08}. Moreover, the region of good
solutions is of interest in the high-SNR analysis, because, for a
given system load, it can be observed that the multiuser
efficiency tends to $1$, consistently with previous theoretical
results~\cite{Tse00}. In what follows, we provide an analysis of
the boundaries of the stable solution regions, as well as their
computationally feasible subregions with practical interest in the
low and high SNR regimes.

\begin{figure} [!htbp]
         \begin{center}
         \includegraphics[width=0.9\columnwidth]{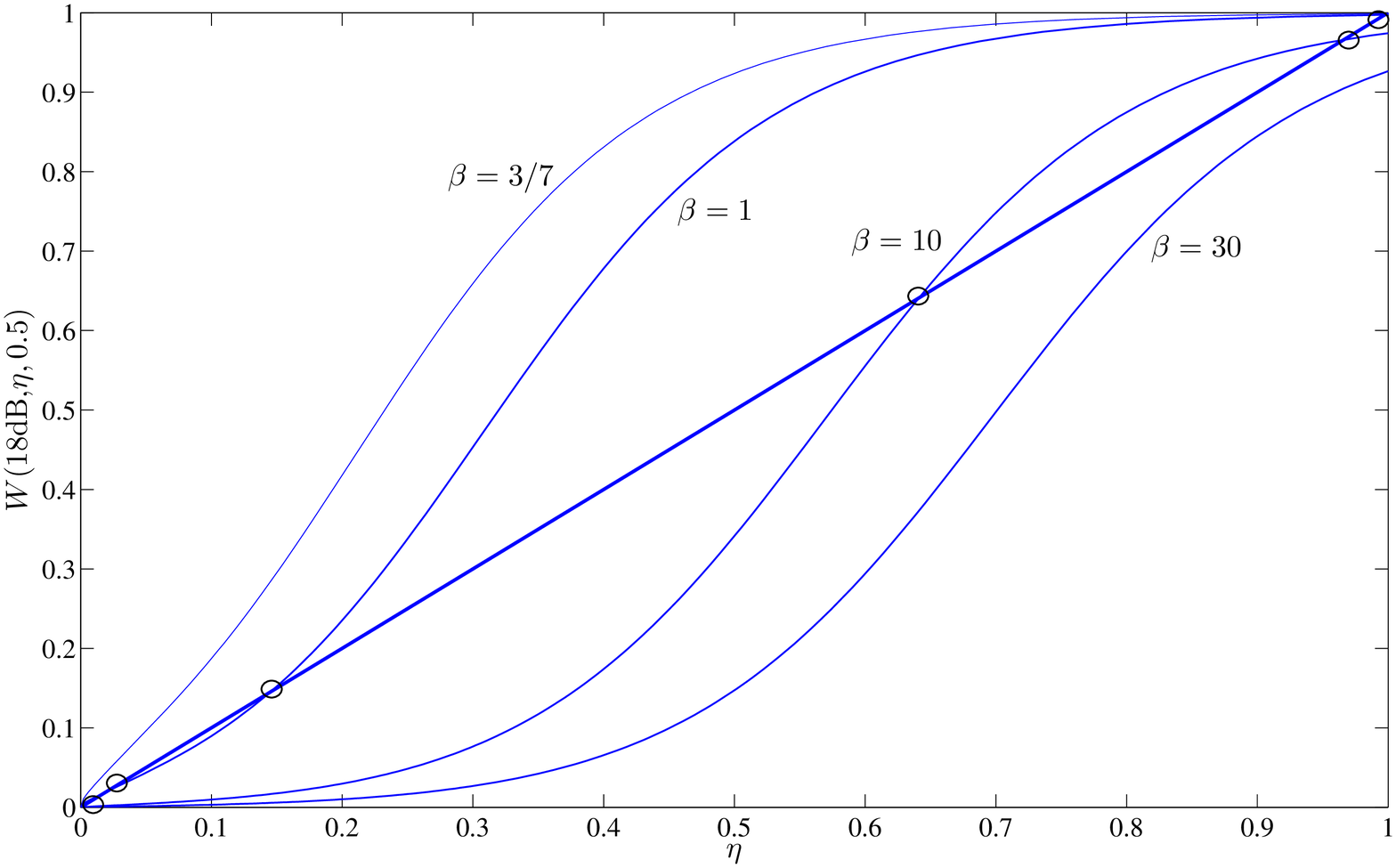}
          \caption{Fixed-point solutions (marked by circles) for different values of $\beta$ and fixed $\alpha=0.5$, and $\gamma=18$ dB.}
           \label{fig4}
            \end{center}
                       \vspace{-7mm}

\end{figure}

A quantitative illustration of the above considerations
is provided by plotting the left- and right-hand sides of~\eqref{eq:2.3} to
obtain fixed points for constant values of amplitude and activity
rate, and as a function of the system load. The solutions
of~\eqref{eq:2.3} are found at the intersection of the curve
corresponding to the right-hand side with the $y=\eta$ line.
Fig.~\ref{fig4} plots different solutions of the right-hand side
of~\eqref{eq:2.3} for increasing system load, $\alpha=0.5$ and
$\gamma=18$ dB:
\begin{equation*}
W(\gamma, \eta, \alpha)\triangleq\frac{1}{1+\beta\gamma
\rm{MMSE}(\eta\gamma,\alpha)}
\end{equation*}
Notice first that the structure of the fixed-point equation in general does not
allow the solution $\eta=0$, and for finite $\gamma$ and
$\beta$, $\eta=1$  is not a solution. In fact, the latter is an
asymptotic solution for large SNR and certain system loads, as the
MMSE decays exponentially to $0$. From Fig.~\ref{fig4}, one can
observe the presence of phase transitions and the coexistence of multiple solutions.
In particular, we observe that for $\beta=3/7$ the good solution is
computationally feasible. On the other hand, for $\beta=1$ and
$\beta=10$ the system has three solutions, where the true solution
belongs to either the bad or the good solution region. When the
system load achieves $\beta=30$, the curve only intersects the
identity curve near $0$, and the operational solution is unique and lies in
a subregion of bad solutions.


\subsection{System load and the space of fixed-point solutions}

Even in the case of good solutions, the multiuser efficiency can be
greatly degraded by the joint effect of the activity rate and the
maximum system load. In order to analyze the fixed-point
equation~\eqref{eq:2.3} from a different perspective and shed light
into the interplay between these parameters, we express the maximum
system load as the following function, derived from~\eqref{eq:2.3}:
\begin{equation}\label{eq:2.9}
    \Upsilon_{\beta}(\gamma,\eta,\alpha)\triangleq\frac{\left(1-\eta\right)}{\eta\gamma \textrm{MMSE}(\eta\gamma,\alpha)}
\end{equation}
Since MMSE is a continuous function of $\eta$ \cite{Guo10}, then
$\Upsilon_{\beta}$ is also a continuous function in any
compact set over the domain $\eta\in(0,1]$ for given SNR and activity
rate. It is also easy to observe that, for small values of $\eta$,
$\Upsilon_{\beta}$ tends to infinity regardless of $\gamma$ and $\alpha$, whereas in
the high-$\eta$ region, which is of interest here, it decays to
$0$. Before analyzing  the behavior of~\eqref{eq:2.9}, we
introduce a few definitions that help describe the boundaries
between the regions with and without coexistence (in the
statistical-physics literature, these boundaries are called
\emph{spinodal lines}~\cite{Tan02}). We also define appropriately
the regions of potentially stable solutions as introduced before.

\begin{definition}
The \emph{critical system load} $\beta^{\star}(\gamma,\alpha)$ is
the maximum load at which a stable good solution
of~\eqref{eq:2.3} exists.
\end{definition}

\begin{definition}
The \emph{transition system load} $\beta_{\star}(\gamma,\alpha)$ is
the minimum load at which the true solution of~\eqref{eq:2.3},
$\eta_{\star}$ coexists with other solutions
$\eta_{\star}'$.
\end{definition}

\begin{definition}
The \emph{good solution region} corresponds to the domain
of~\eqref{eq:2.9} formed by the maximum $\eta$  in every set of
pre-images of $\Upsilon_{\beta}$ below the critical system load:
\begin{equation}\label{eq:2.9b}
\mathcal{R}_g=\left\{\eta\in[0,1],
\eta=\max\{\Upsilon_{\beta}^{-1}(\beta)\}, \forall \beta \in
[0,\beta^{\star}]\right\}
\end{equation}
Similarly, the \emph{bad solution region} corresponds to the
domain of \eqref{eq:2.9} formed by the minimum $\eta$  in every
set of pre-images of $\Upsilon_{\beta}$ above the transition
system load:
\begin{equation}\label{eq:2.9c}
\mathcal{R}_b=\left\{\eta\in[0,1],
\eta=\min\{\Upsilon_{\beta}^{-1}(\beta)\}, \forall \beta \in
[\beta_{\star},+\infty)\right\}
\end{equation}

\end{definition}

\begin{figure} [!htbp]
         \begin{center}
         \includegraphics[width=0.9\columnwidth]{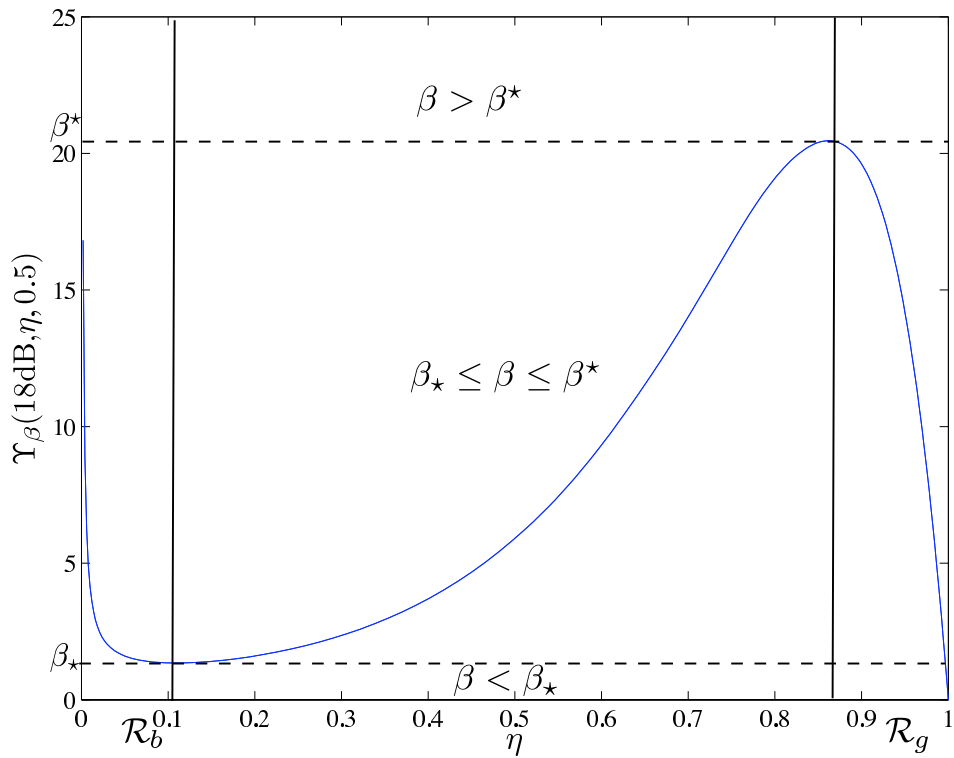}%
          \caption{System load function in the multiuser efficiency domain for $\alpha=0.5$ and $\gamma=18$ dB.}
            \label{fig5b}
                       \vspace{-4mm}

            \end{center}
\end{figure}


\bigskip
Fig. \ref{fig5b} illustrates $\Upsilon_{\beta}$ (for fixed SNR and activity
rate) and show the regions defined by the aforementioned parameters.
It is important to remark that both system loads defined above
delimit the regions from which there is phase coexistence
($\beta_{\star}\leq\beta\leq\beta^{\star}$) from the areas where
there is one solution ($\beta>\beta^{\star}$ or
$\beta<\beta_{\star}$). Additionally, Fig. 
\ref{fig5b} illustrates the set of solutions that satisfy conditions
\eqref{eq:2.9b}, \eqref{eq:2.9c}.



Fig. \ref{fig5b} illustrates that it seems useful to define
analytically the domain where stable solutions can be found.
Beforehand, we differentiate for convenience the case with
unknown number of users $\alpha\in (0,1)$, from the case where all
users are active ($\alpha=1$). We do not consider the case
$\alpha=0$.

\subsubsection{Case $\alpha\in(0,1)$}

In order to analyze the conditions on the system load, SNR, and
activity rate, for which we can find a good solution, we use the
asymptotic results on the MMSE \eqref{eq:2.6},
 yielding lower and upper bounds
$L(.)\leq\Upsilon_{\beta}(\gamma,\eta,\alpha)\leq U(.)$ for large enough $\eta\gamma$,
where
\begin{eqnarray}
L(\gamma,\eta,\alpha)&\triangleq&\frac{\left(1-\eta\right)}{\sqrt{\pi\eta\gamma\alpha(1-\alpha)}}
e^{\eta\gamma/8}\label{eqaux1}\\
U(\gamma,\eta,\alpha)&\triangleq&\frac{\left(1-\eta\right)}{2}\sqrt{\frac{\pi}{\eta\gamma\alpha(1-\alpha)}}
e^{\eta\gamma/8}\label{eqaux2}
\end{eqnarray}


Although not exact for low SNR, the dependence on $\eta$ of the
upper and lower bound provides a good approximation for the
dependence of $\Upsilon_{\beta}$ for large SNR
and given $\alpha$. Hence, by using $U(.)$ and $L(.)$, we
obtain necessary and sufficient conditions  that determine the
regions of stable solutions and provide analytical expressions for
the transition and critical system loads. The main result for
$\alpha\in(0,1)$ follows:

\begin{thm}\label{thm4}
Given the range of activity rates $\alpha\in(0,1)$, a necessary
condition for phase coexistence is
\begin{equation}\label{eq:2.10}
\gamma\geq4(3+2\sqrt{2}), \quad(\gamma\geq 13.67 \; \textrm{dB})
\end{equation}
Moreover, for high SNR, the condition is met and the transition
system load is bounded by
\begin{equation}\label{eq:2.11}
L(\gamma,\eta_m,\alpha)<\beta_{\star}(\gamma,\alpha)<U(\gamma,\eta_m,\alpha)
\end{equation}
while the critical system load is bounded by
\begin{equation}\label{eq:2.12}
L(\gamma,\eta_M,\alpha)<\beta^{\star}(\gamma,\alpha)<U(\gamma,\eta_M,\alpha)
\end{equation}
and $\eta_{m}$, $\eta_{M}$ are given by
\begin{eqnarray*}
\eta_m&\triangleq&(\gamma/2-2-4\Delta(\gamma))/\gamma\\
\eta_M&\triangleq&(\gamma/2-2+4\Delta(\gamma))/\gamma
\end{eqnarray*}
where $\Delta(\gamma)=\sqrt{(\gamma/8)^2-3\gamma/8+1/4}$.

Hence, the bad-solution region is given by
$\mathcal{R}_{b}=(0,\eta_{m}]$, whereas the good-solution region
is $\mathcal{R}_{g}=[\eta_{M},1]$. Similarly, the subregions of
single bad solutions, that we shall denote
$R_{bc}=(0,\eta_{bc})\subset\mathcal{R}_{b}$,
and of single good solutions, denoted by
$R_{gc}=(\eta_{gc},1]\subset\mathcal{R}_{g}$,
satisfy
\begin{eqnarray*}
\eta_{bc}&=&\min\{\Upsilon_{\beta}^{-1}(\beta^{\star})\}>\eta^{\star}_{bc}\\
\eta_{gc}&=&\max\{\Upsilon_{\beta}^{-1}(\beta_{\star})\}<\eta^{\star}_{gc}
\end{eqnarray*}
where $\eta^{\star}_{bc}\triangleq\min\{U^{-1}(\beta^{\star})\}$,
and $\eta^{\star}_{gc}\triangleq\max\{L^{-1}(\beta^{\star})\}$ are
obtained from the bounds.
\end{thm}

\begin{proof}
See Appendix \ref{app:beta_proof}.
\end{proof}

\bigskip
The above result provides the general boundaries of the space of
solutions of our problem. It is important to note that $\eta_m$ and
$\eta_M$ are very good approximations for high SNR of the positions
of the minimum and maximum observed in Fig. \ref{fig5b}, which 
determine the transition and the critical system loads. As a consequence, remark that Theorem \ref{thm4} analytically  tells us 
the range of $\beta$'s for which 
there are either single or multiple solutions based on the up-to-a-constant approximation of  $\Upsilon_{\beta}$ by \eqref{eqaux1} and \eqref{eqaux2}. Similarly,
$\eta^{*}_{bc}$ and $\eta^{*}_{gc}$ are tight bounds of the
boundaries of the single-solution regions as $U(.)$ and
$L(.)$ are of $\Upsilon_{\beta}(.)$.
Note also that the
activity rate affects the boundaries in the same symmetrical manner
as it does the MMSE (i.e., the worst case here also corresponds to
$\alpha=0.5$) but has no impact on the operational region, that is
only reduced in size by increasing
the SNR. 
In particular, these regions are characterized, in the limit of high
SNR, as follows:
\begin{cor}
In the limit of high SNR, $\mathcal{R}_{g}\to \{1\}$,
$\mathcal{R}_{b}\to \{0\}$, and consequently $\mathcal{R}_{gc}\to
\{1\}$, and $\mathcal{R}_{bc}\to \{0\}$ .
\end{cor}
\begin{proof}
The above corollary results from
\begin{eqnarray*}
\lim_{\gamma\to\infty}\eta_M&=&\lim_{\gamma\to\infty}\frac{(\gamma/2-2+4\sqrt{(\gamma/8)^2-3\gamma/8+1/4})}{\gamma}=1 \\
\lim_{\gamma\to\infty}\eta_m&=&\lim_{\gamma\to\infty}\frac{(\gamma/2-2-4\sqrt{(\gamma/8)^2-3\gamma/8+1/4})}{\gamma}=0
\end{eqnarray*}
\end{proof}

\bigskip
Note that, given a system load $\beta$ with $\beta^{\star}>\beta$,
for sufficiently large SNR the unique true (large-system) solution
is $\eta=1$, which corroborates the main result in~\cite{Tse00}.
Moreover, the description of the feasible good solutions by
analytical means allows the computation of a sufficient condition
on the system load to guarantee a given multiuser efficiency in
practical implementations. More specifically, we use the
aforementioned lower bound on $\Upsilon$ to state that any system
load below $L(.)$ guarantees that the given multiuser efficicency is achieved. The
result is stated as follows:
\begin{cor}\label{cor1}
The maximum system load, $\beta_{\alpha,\eta}$, for a given
activity rate and multiuser efficiency requirement,
$\eta=1-\epsilon$, where $0<\epsilon\ll1$, that lies in $R_{gc}$, is
lower-bounded in the high-SNR region by
\begin{equation}\label{eq:2.13}
\beta_{\alpha,\eta}>\frac{\epsilon}{\sqrt{\pi\eta\gamma\alpha(1-\alpha)}}
e^{(1-\epsilon)\gamma/8}.
\end{equation}
\end{cor}

In Fig.~\ref{fig6} we show the numerical values of the transition
and the critical system load as a function of the SNR in the
$(\gamma,\beta)$ space. We also use the asymptotic expansion to
derive upper and lower bounds, respectively. The plotted curves are
the spinodal lines, which mark the boundary between the regions with
and without solution coexistence. The $\beta_{\star}$ (lower branch)
separates the region where the bad solution disappears, whereas
$\beta^{\star}$ (upper branch) contains the bifurcation points at
which the operational solution disappears. The intersection point
between both branches corresponds to the SNR
threshold~\eqref{eq:2.10}, which provides the necessary condition
for solution coexistence.


\begin{figure} [!htbp]
         \begin{center}
         \includegraphics[width=1.0\columnwidth]{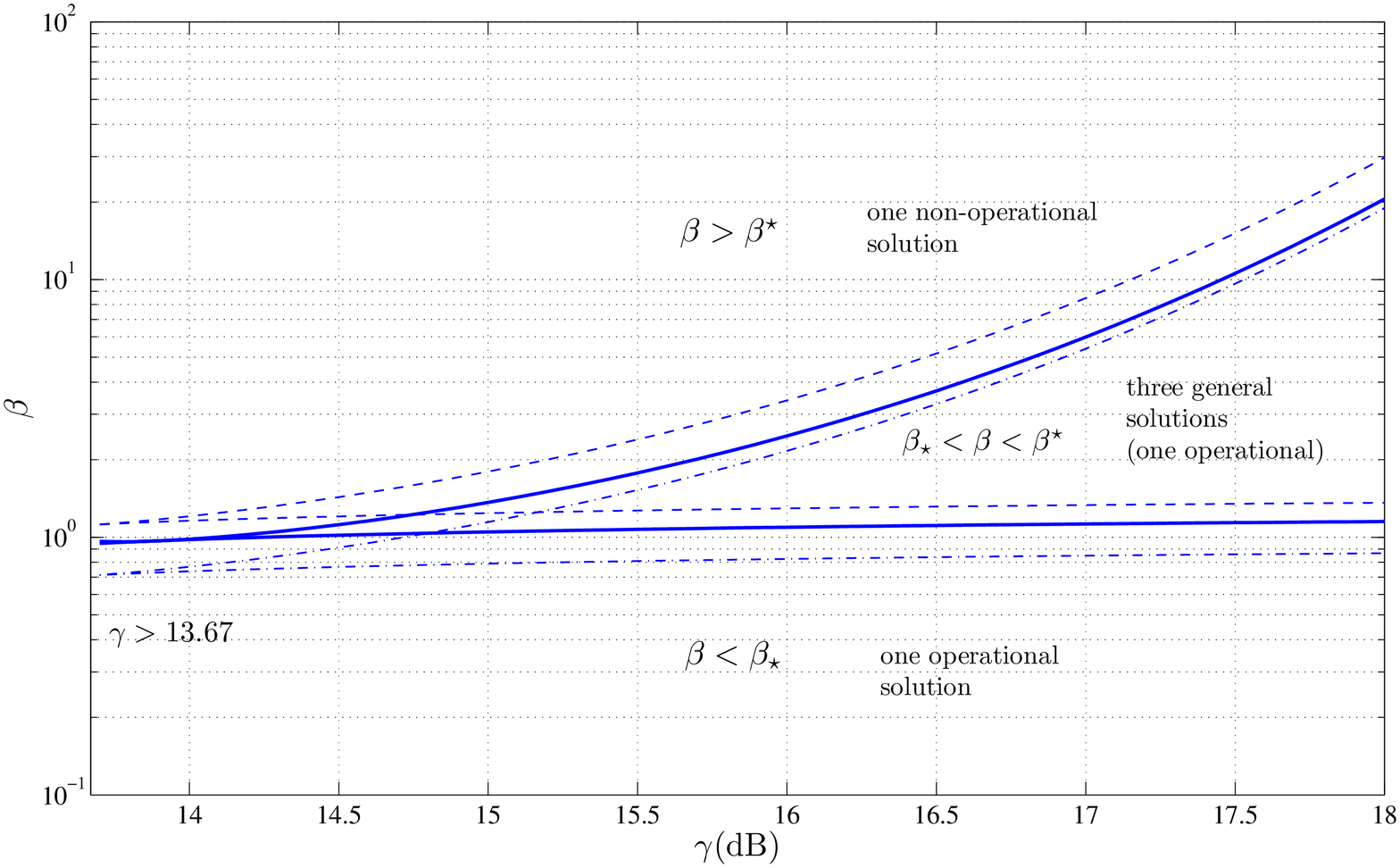}
          \caption{\label{fig6}
            {Upper and lower bounds on the numerical spinodal lines (thick line) for $\alpha=0.5$. }}
            \end{center}
                       \vspace{-7mm}

\end{figure}

\subsubsection{Case $\alpha=1$}
We now apply the same reasoning for the ``classical'' approach to
multiuser detection, corresponding to activity rate $1$. In this
case, using the approximation in~\cite{Loz06}, the system load
function can be lower-bounded by
\begin{equation}\label{eq:2.14}
   \Upsilon_{\beta}(\gamma, \eta, 1)=\frac{\left(1-\eta\right)}{\eta\gamma
    \textrm{MMSE}(\eta\gamma,1)}<\frac{\left(1-\eta\right)e^{\eta\gamma/2}}{\sqrt{\pi\eta\gamma}}
\end{equation}
Hence, we can derive the following spinodal lines
\begin{cor}
Given $\alpha=1$ a necessary condition for the phase coexistence is
that
\begin{equation*}
\gamma\geq3+2\sqrt{2}, \quad(\gamma\geq 7.65\textrm{dB}).
\end{equation*}

Moreover, for high SNR, the condition is met and the transition
system load is upper-bounded by
\begin{equation}\label{eq:2.15}
\beta_{\star}<\frac{\left(1-\eta_{m'}\right)}{2}\sqrt{\frac{\pi}{\eta_{m'}\gamma}}
e^{\frac{\eta_{m'}\gamma}{2}}
\end{equation}
and the critical system load is upper-bounded by
\begin{equation}\label{eq:2.16}
\beta^{\star}<\frac{\left(1-\eta_{M'}\right)}{2}\sqrt{\frac{\pi}{\eta_{M'}\gamma}}
e^{\frac{\eta_{M'}\gamma}{2}}
\end{equation}
where $\eta_{m'}$ and $\eta_{M'}$ are given by
\begin{eqnarray*}
\eta_{m'}&\triangleq&\frac{(\gamma/2-1/2-\Lambda(\gamma))}{\gamma}\\
\eta_{M'}&\triangleq&\frac{(\gamma/2-1/2+\Lambda(\gamma))}{\gamma}
\end{eqnarray*}
and $\Lambda(\gamma)=\sqrt{(\gamma/2)^2-3\gamma/2+1/4}$. 

Hence,
the bad solution region is given by $\mathcal{R}_{b}=(0,\eta_{m'}]$ whereas
the good solution region is $\mathcal{R}_{g}=[\eta_{M'},1]$.
\end{cor}

\begin{proof}
The proof is analogous to that of Theorem \ref{thm4}.
\end{proof}

The same consequence for the asymptotic operational region holds
here.
\begin{cor}
In the limit of high SNR, $\mathcal{R}_{g}\to \{1\}$, and $\mathcal{R}_{b}\to \{0\}$.
\end{cor}
\begin{proof}
This corollary results from
\begin{eqnarray*}
\lim_{\gamma\to\infty}\eta_M=\lim_{\gamma\to\infty}\frac{(\gamma/2-1/2+\sqrt{(\gamma/2)^2-3\gamma/2+1/4})}{\gamma}=1 \\
\lim_{\gamma\to\infty}\eta_m=\lim_{\gamma\to\infty}\frac{(\gamma/2-1/2-\sqrt{(\gamma/2)^2-3\gamma/2+1/4})}{\gamma}=0
\end{eqnarray*}
\end{proof}
\begin{figure} [!htbp]
         \begin{center}
         \includegraphics[width=0.9\columnwidth]{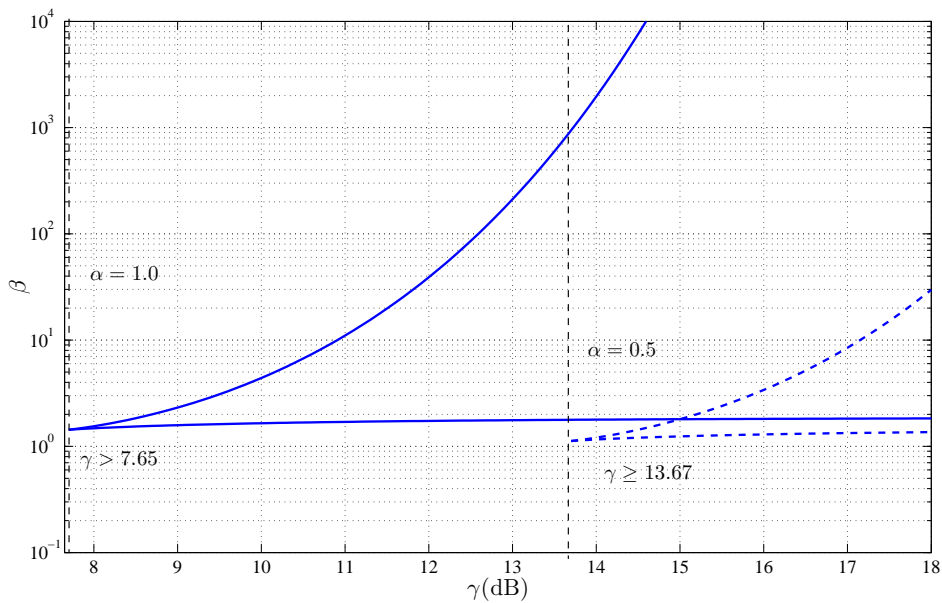}
          \caption{\label{fig7}
            {Comparison of upper bounds on the spinodal lines for $\alpha=1.0$ (left) and $\alpha=0.5$ (right).}}
                       \vspace{-7mm}

            \end{center}
\end{figure}

In Fig.~\ref{fig7}, one can observe a $6$~dB-difference between the
spinodal lines corresponding to $\alpha=0.5$ and to $\alpha=1.0$.
This is due to the minimum distance of the underlying constellations, 
which causes the MMSE to have different exponential decays. This
can be interpreted by saying that the addition of activity detection
to data detection is reflected by a $6$~dB increase of the SNR
needed to achieve the same system load performance. Moreover, with
$\alpha <1$, the transition system load is lower than the case where
all users are active, and, therefore, computationally good solutions
correspond to lower values of the maximum system load.

\subsection{Maximum system load with error probability constraints}
A natural application of the above results to practical designs
appears when the quality-of-service requirements of the system are
specified in terms of uncoded error probability. Such an
application can provide some extra insight into the plausible
values of $\beta$ with joint activity and data detection for
efficient design of large CDMA systems. Once a
multiuser-efficiency requirement is assigned, the corresponding
probability of error follows naturally. Note first that, in order
to detect the activity as well as the transmitted data, our model
deals with a ternary constellation $\{-1,0,1\}$. When any of these
symbols is transmitted by each user with constant SNR$=\gamma$
through a bank of large-system equivalent white Gaussian noise channels with
variance $1/\eta$, the probability of error over $\Xc$
depends on the prior probabilities as well as the Euclidean
distance between the symbols. The error probability implied by the replica analysis is
\begin{equation}\label{eq:2.17}
P_e(\eta,\gamma,\alpha)=2(1-\alpha)Q\left(\frac{\sqrt{\eta\gamma}}{2}+\frac{\lambda_{\alpha}}
{\sqrt{\eta\gamma}}\right)+\alpha
Q\left(\frac{\sqrt{\eta\gamma}}{2}-\frac{\lambda_{\alpha}}
{\sqrt{\eta\gamma}}\right)
\end{equation}
where $Q(x)\triangleq
\frac{1}{\sqrt{2\pi}}\int_{x}^{\infty}e^{-\frac{t^2}{2}}dt$ is the
Gaussian tail function, and
$\lambda_{\alpha}\triangleq\ln\left(\frac{2(1-\alpha)}{\alpha}\right)$.

The relationship between $\eta$ and $P_e$ for our particular case
can be used to reformulate the bounds on the function
$\Upsilon_{\beta}$ in terms of error probability.
\begin{cor}
The maximum system load, $\Upsilon_{\beta}(\eta,\gamma,\alpha)$, for
a given error probability $P_e$, $\gamma$, and activity rate is bounded for high SNR by:
\begin{equation}\label{eq:2.18}
L(\gamma,\eta_{\textrm{max}},\alpha)<\Upsilon_{\beta}(\eta_{\textrm{max}},\gamma,\alpha)<U(\gamma,\eta_{\textrm{max}},\alpha)
\end{equation}
where
\begin{equation*}
\eta_{\textrm{max}}\triangleq\max\left\{\eta_P,\eta_{gc}\right\},
\end{equation*}
and $(\eta_P, \gamma, \alpha)$ is the pre-image of $P_e$.  

\end{cor}

\begin{proof}
The result is obtained by noticing that the multiuser efficiency
requirement extracted from $P_e$ must lie on the subregion $(\eta_{gc},1]$.
\end{proof}

\bigskip
Notice that, if the error probability satisfies
$\eta_{p}\leq\eta_{gc}=\eta_{\textrm{max}}$, then the constraint is
described by the bounds on the transition load \eqref{eq:2.12}. However, if
$\eta_{gc}<\eta_{p}=\eta_{\textrm{max}}$, then, for
Corollary~\ref{cor1}, the maximum system load can be also easily
bounded. Fig.~\ref{fig8} plots the critical system load for
two different error probabilities requirements and three
different activity rates.

\begin{figure} [!htbp]
         \begin{center}
         \includegraphics[width=0.9\columnwidth]{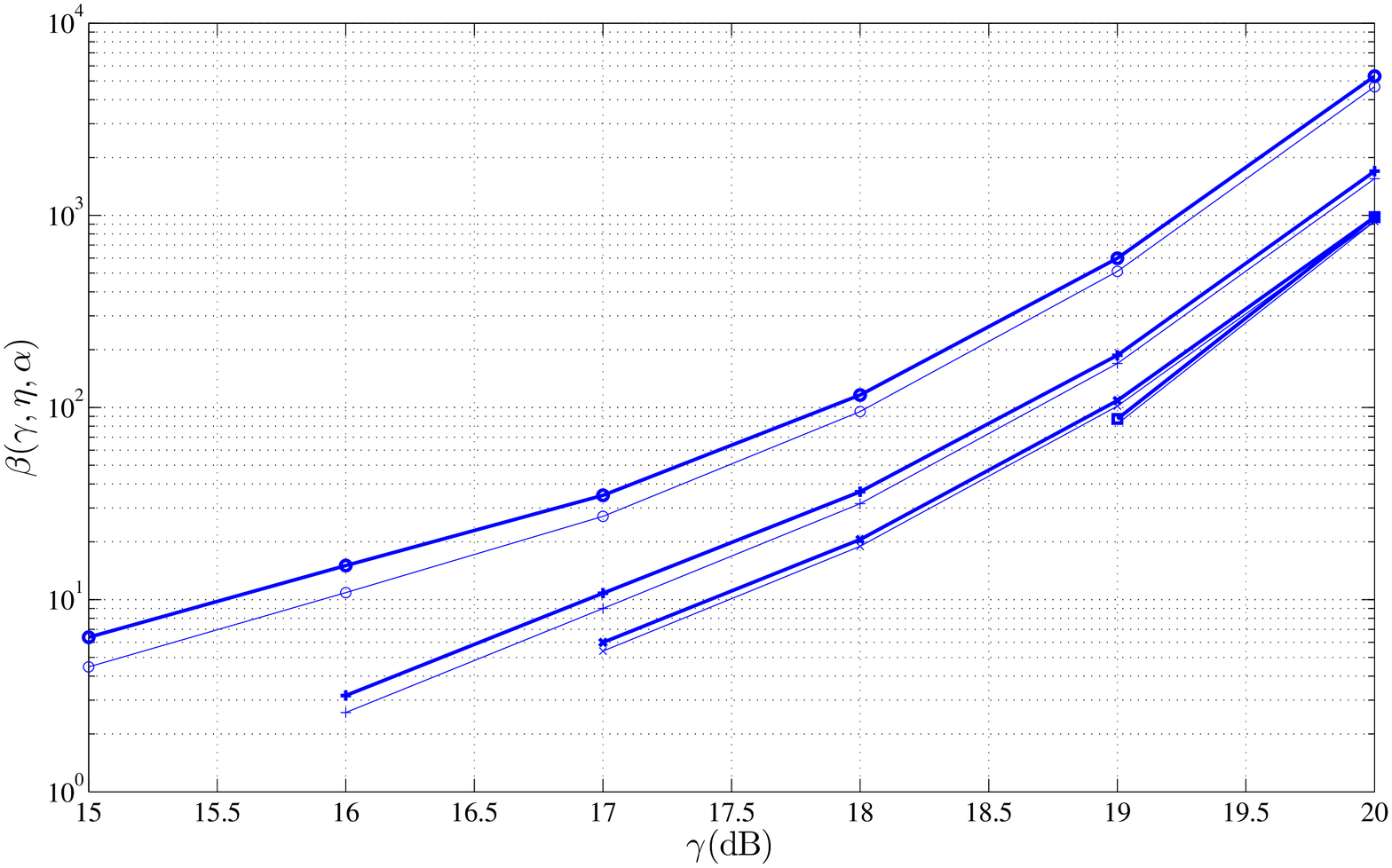}
          \caption{\label{fig8}
            {Critical system load for different uncoded error probabilities and activity rates. Thicker lines represent
            numerical results, whereas regular lines show the corresponding lower bounds.
            Lines with circle markers:
            $P_e=10^{-3}$ and $\alpha=0.99$. Lines with cross markers:
            $P_e=10^{-3}$ and
            $\alpha=0.1$. Lines with star markers:
            $P_e=10^{-3}$ and $\alpha=0.5$. Lines with square markers: error probability
            $10^{-5}$ and $\alpha=0.5$. }}
            \end{center}
                       \vspace{-7mm}

\end{figure}

\section{Conclusions}
\label{section:conclusions}
%
We have analyzed multiuser detectors for CDMA where the fraction
of active users is unknown, and must be estimated in a tracking
phase. Using a large-system approach and statistical-physics
tools, we have derived a fixed-point equation for the optimal
user-and-data detector, and provided asymptotic bounds for the
corresponding MMSE. Further, we have described the space of
stable solutions of the fixed-point equation, and derived explicit
bounds on the critical and transition system loads for all users'
activity rates. These are consistent with the results obtained
under the classic multiuser-detection assumption ($\alpha=1.0$)
made in the literature. The study of the so-called spinodal
lines allowed us to determine the regions of stable good and bad
solutions, including subregions of single solutions (also
computationally feasible), in the system load vs. SNR parameter space
of our model. Our results show that for a user-and-data detector,
the boundaries of the space of solutions do depend on the activity
rate, whereas the regions of stable solutions (good and bad) are
only affected by the SNR. Hence, the overall system load
performance keeps a symmetric behavior with respect to $\alpha$.
In practical implementations with high quality-of-service demands,
we are interested in maximizing the critical system load, while
keeping the optimal detector in the feasible subregion of good
multiuser efficiencies, such that a a wider range of potential users
can successfully access the channel with a given rate. By
increasing the SNR, this goal can be achieved, but for limited
SNR, the certainty on the users' activity allows allocation of
more users for a given spreading length. A relevant example
corresponds to a system with a given error
probability requirement. For this case, we have shown that, for sufficiently
large SNR, we can choose the minimum multiuser efficiency in the
domain of feasible good solutions, and maximize the critical
system load regardless of the error probability target.

One of the assumptions of this paper is to model the activity as an i.i.d. process.
Further extensions of this work to non-i.i.d scenarios can be found in \cite{Tau082} where
the users' activity evolves according to a Markov process and
in \cite{TauGui10}, where users transmit encoded messages and
the activity is correlated over the coded blocks of each user.

\vspace{-3mm}

\appendices
\section{Proof of Corollary \ref{thm1}}
\label{app:proof_mmse}

We first derive the MMSE for our particular ternary input distribution in
the general fixed-point equation \eqref{eq:1.13}:
\begin{align*}
\textrm{MMSE}(\eta\gamma,\alpha)&=\EE\left[\biggl(\b^k-\EE\{\b^k|\S,y\}\biggl)^{2}\right]\\
&=\alpha- \int\frac{\EE_{\b^k}^2\{\b^kP(y|\eta,
\b^k, S)\}}{\EE_{\b^k}\{P(y|\eta, \b^k,
S)\}}\textrm{d}y\\
&=\alpha-\int\frac{\left(\sqrt{\frac{\eta}{2\pi}}\frac{\alpha}{2}\left[e^{-\frac{\eta}{2}(y-\sqrt{\gamma})^2}-
e^{-\frac{\eta}{2}(y+\sqrt{\gamma})^2}\right]\right)^2}{\sqrt{\frac{\eta}{2\pi}}\left(\frac{\alpha}{2}
\left[e^{-\frac{\eta}{2}(y-\sqrt{\gamma})^2}+
e^{-\frac{\eta}{2}(y+\sqrt{\gamma})^2}\right]\right)+(1-\alpha)e^{\frac{\eta}{2}y^2}}\textrm{d}y\\
&=\alpha-\frac{1}{2}\sqrt{\frac{\eta}{2\pi}}\left[\int\frac{e^{-\frac{\eta}{2}(y-\sqrt{\gamma})^2}
\alpha^{2}\sinh(\eta y\sqrt{\gamma})}{\alpha\cosh(\eta y
\sqrt{\gamma})+(1-\alpha)e^{\frac{\eta \gamma}{2}}}\textrm{d}y-
\int\frac{e^{-\frac{\eta}{2}(y'+\sqrt{\gamma})^2}
\alpha^{2}\sinh(\eta y'\sqrt{\gamma})}{\alpha\cosh(\eta y'
\sqrt{\gamma})+(1-\alpha)e^{\frac{\eta \gamma}{2}}}\textrm{d}y'\right].
\end{align*}

After appropriate change of variables 

%
the MMSE is
\begin{align*}
\rm{MMSE}(\eta\gamma,\alpha)
&=\alpha-\int\frac{1}{\sqrt{2\pi}}e^{\frac{-y^2}{2}}\frac{\alpha^2\sinh[\eta
\gamma-y\sqrt{\eta\gamma}]}{\alpha\cosh[\eta\gamma-y\sqrt{\eta
\gamma}]+(1-\alpha)e^{\eta\gamma/2}}\textrm{d}y
\end{align*}
Since the SNR is constant among users, it follows naturally that
the large-system fixed-point equation is given by \eqref{eq:2.3}.

\section{Proof of Theorem \ref{thm2}}
\label{app:proof_mmse_limits}

From now on, we omit the explicit indication
of the arguments of the MMSE function. The lower bound is obtained
by noting that, for large SNR, the general MMSE, denoted
$\textrm{MMSE}_{\alpha}$, is lower-bounded by
$\textrm{MMSE}_{\{0,1\}}$, which describes the detection
performance when the transmitted symbols are $\{0,1\}$, with
probabilities $\{1-\alpha, \alpha\}$. In this case the MMSE has
the following form
\begin{align*}
\textrm{MMSE}_{\{0,1\}}&=\alpha-\int\frac{\left(\sqrt{\frac{\eta}{2\pi}}\alpha
e^{-\frac{\eta}{2}(y-\sqrt{\gamma})^2}\right)^2}
{\sqrt{\frac{\eta}{2\pi}}\left(\alpha
e^{-\frac{\eta}{2}(y-\sqrt{\gamma})^2}+
(1-\alpha)e^{-\frac{\eta}{2}y^2}\right)}\textrm{d}y
\end{align*}
where
$\lambda_{\alpha}\triangleq\frac{1}{2}\ln\left(\frac{1-\alpha}{\alpha}\right)$.
After some manipulation and appropriate changes of variables, 
we obtain
\begin{align*}
\textrm{MMSE}_{\{0,1\}}&=\alpha-\alpha
e^{-\frac{3\eta\gamma}{8}-\lambda_{\alpha}}\sqrt{\frac{1}{2\pi}}\int_{-\infty}^{\infty}
e^{-\frac{1}{2}(2z-\sqrt{\eta\gamma}+\frac{2\lambda_{\alpha}}{\sqrt{\eta\gamma}})^2}\sech\left(
z\sqrt{\eta\gamma}\right)\textrm{d}z.
\end{align*}
Use now the following asymptotic expansion for $\sech(z)$:
$\abs{z}\to\infty$
\begin{equation*}
\sech(z)=
2e^{-\abs{z}}\left(1+\sum_{\ell=1}^{\infty}(-1)^{\ell}e^{-2\ell\abs{z}}\right)
\end{equation*}
and obtain
\begin{align*}
\textrm{MMSE}_{\{0,1\}}=\alpha-\alpha
e^{-\frac{3\eta\gamma}{8}-\lambda_{\alpha}}\sqrt{\frac{1}{2\pi}}\Biggl[&\int_{-\infty}^{0}
e^{-\frac{1}{2}(2z-\sqrt{\eta\gamma}+\frac{2\lambda_{\alpha}}{\sqrt{\eta\gamma}})^2}
2e^{z\sqrt{\eta\gamma}}\left(1+\sum_{\ell=1}^{\infty}(-1)^{\ell}e^{2\ell
z\sqrt{\eta\gamma}}\right)\textrm{d}z\\
+&\int_{0}^{\infty}e^{-\frac{1}{2}(2z-\sqrt{\eta\gamma}+\frac{2\lambda_{\alpha}}{\sqrt{\eta\gamma}})^2}
2e^{-z\sqrt{\eta\gamma}}\left(1+\sum_{\ell=1}^{\infty}(-1)^{\ell}e^{-2\ell
z\sqrt{\eta\gamma}}\right)\textrm{d}z\Biggl].
\end{align*}

We now use the expansion
\begin{align*}
\textrm{MMSE}_{\{0,1\}}=\alpha-\alpha
e^{-\eta\gamma/8+\lambda_{\alpha}-\frac{2\lambda^2_{\alpha}}{\eta\gamma}}\Biggl[&\sum_{\ell=0}^{\infty}(-1)^{\ell}\sqrt{\frac{2}{\pi}}\int_{-\infty}^{0}
e^{-2z^2+\left((3+2\ell)\sqrt{\eta\gamma}-\frac{4\lambda_{\alpha}}{\sqrt{\eta\gamma}}\right)z}\textrm{d}z\\
+&\sum_{\ell=0}^{\infty}(-1)^{\ell}\sqrt{\frac{2}{\pi}}\int_{0}^{\infty}e^{-2z^2+\left((1-2\ell)\sqrt{\eta\gamma}-\frac{4\lambda_{\alpha}}{\sqrt{\eta\gamma}}\right)z}
\textrm{d}z\Biggl].
\end{align*}
Express now the integrals in terms of the Q-function,
$Q(x)\triangleq\frac{1}{\sqrt{2\pi}}\int_{x}^{\infty}e^{-\frac{t^2}{2}}dt$
\begin{align*}
\textrm{MMSE}_{\{0,1\}}=\alpha
e^{-\eta\gamma/8+\lambda_{\alpha}-\frac{2\lambda^2_{\alpha}}{\eta\gamma}}
&\Biggl[
e^{\frac{\left(\frac{\sqrt{\eta\gamma}}{2}-\frac{2\lambda_{\alpha}}{\sqrt{\eta\gamma}}\right)^{2}}{2}}Q\left(\frac{\sqrt{\eta\gamma}}{2}-\frac{2\lambda_{\alpha}}{\sqrt{\eta\gamma}}\right)
-e^{\frac{\left(\frac{3\sqrt{\eta\gamma}}{2}-\frac{2\lambda_{\alpha}}{\sqrt{\eta\gamma}}\right)^{2}}{2}}Q\left(\frac{3\sqrt{\eta\gamma}}{2}-\frac{2\lambda_{\alpha}}{\sqrt{\eta\gamma}}\right)\\
&-\sum_{\ell=1}^{\infty}(-1)^{\ell}\biggl[
e^{\frac{\left(\frac{(2\ell-1)\sqrt{\eta\gamma}}{2}+\frac{2\lambda_{\alpha}}{\sqrt{\eta\gamma}}\right)^{2}}{2}}
Q\left(\frac{(2\ell-1)\sqrt{\eta\gamma}}{2}+\frac{2\lambda_{\alpha}}{\sqrt{\eta\gamma}}\right)\\
&+e^{\frac{\left(\frac{(3+2\ell)\sqrt{\eta\gamma}}{2}-\frac{2\lambda_{\alpha}}{\sqrt{\eta\gamma}}\right)^{2}}{2}}
\left(Q\left(\frac{(3+2\ell)\sqrt{\eta\gamma}}{2}-\frac{2\lambda_{\alpha}}{\sqrt{\eta\gamma}}\right)\right)\biggl]\Biggl].
\end{align*}
Next, use the expansion of the Q function, \cite{Ver98}:
\begin{equation}\label{ap1}
Q(x)=\frac{e^{-x^2/2}}{\sqrt{2\pi}x}\left(1+\sum_{\ell=1}^{\infty}(-1)^{\ell}
\frac{\prod_{q=1}^{\ell}(2q-1)}{x^{2\ell}}\right)
\end{equation}
to obtain
\begin{align*}
\textrm{MMSE}_{\{0,1\}}=2\sqrt{\frac{\alpha(1-\alpha)}{2\pi\eta\gamma}}
e^{-\eta\gamma/8-\frac{2\lambda^2_{\alpha}}{\eta\gamma}}
\Biggl[&\frac{1}{\left(1-\frac{4\lambda_{\alpha}}{\eta\gamma}\right)}
+\frac{1}{\left(1+\frac{4\lambda_{\alpha}}{\eta\gamma}\right)}-
\frac{1}{\left(3-\frac{4\lambda_{\alpha}}{\eta\gamma}\right)}-
\frac{1}{\left(3+\frac{4\lambda_{\alpha}}{\eta\gamma}\right)}\\
&+\dots +\mathcal{O}\left(\frac{1}{\sqrt{\eta\gamma}}\right)\Biggl]
\end{align*}
where the linear term in $\lambda_{\alpha}$ is substituted in the
common factor. Assuming a large value $\eta\gamma$, and using the
result
\begin{align*}
2\sum_{n=0}^{\infty}\frac{(-1)^{n+1}}{2n+1}=\frac{\pi}{2}
\end{align*}
we obtain the following lower bound
\begin{align*}
\textrm{MMSE}_{\{0,1\}}>2\sqrt{\frac{\alpha(1-\alpha)}{2\pi\eta\gamma}}
e^{-\eta\gamma/8}& \sqrt{2}
=2\sqrt{\frac{\alpha(1-\alpha)}{\pi\eta\gamma}}
e^{-\eta\gamma/8}.
\end{align*}


As far as the upper-bound is concerned, we derive the general MMSE,
denoted $\textrm{MMSE}_{\alpha}$, and its particular case
when all users are assumed to be active, denoted $\textrm{MMSE}_{1}$.
Hence, we express the corresponding integrals in an analogous
manner
\begin{eqnarray*}
\zeta_{\alpha}&=&\int\frac{1}{\sqrt{2\pi}}e^{\frac{-(y-\sqrt{\eta\gamma})^2}{2}}\frac{\alpha\sinh(y\sqrt{\eta\gamma})}{\alpha\cosh(y\sqrt{\eta
\gamma})+(1-\alpha)e^{\eta\frac{\gamma}{2}}}\textrm{d}y.\\
\zeta_{1}&=&\int\frac{1}{\sqrt{2\pi}}e^{\frac{-(y-\sqrt{\eta\gamma})^2}{2}}\tanh(y\sqrt{\eta\gamma})\textrm{d}y.
\end{eqnarray*}
We now obtain
\begin{align*}
\textrm{MMSE}_{\alpha}&=\alpha(1-\zeta_{\alpha})=\alpha\left(1+\zeta_{1}-\zeta_{1}-\zeta_{\alpha}\right)=\alpha\left(
(1-\zeta_{1})+(\zeta_{1}-\zeta_{\alpha})\right)\\
&=\alpha\left(\textrm{MMSE}_{1}+(\zeta_{1}-\zeta_{\alpha})\right)
\end{align*}
Next, we expand $\alpha(\zeta_{1}-\zeta_{\alpha})$, which yields
\begin{align}
\alpha(\zeta_{1}-\zeta_{\alpha})&=(1-\alpha)e^{\eta\gamma/2}
\int\frac{1}{\sqrt{2\pi}}e^{\frac{-(y-\sqrt{\eta\gamma})^2}{2}}\frac{\alpha\sinh[y\sqrt{\eta\gamma}]}{\alpha\cosh^2[y\sqrt{\eta
\gamma}]+(1-\alpha)e^{\eta\gamma/2}\cosh(y\sqrt{\eta
\gamma})}\textrm{d}y\notag\\
&=(1-\alpha)e^{\eta\gamma/2}\int\frac{1}{\sqrt{2\pi}}e^{\frac{-(y-\sqrt{\eta\gamma})^2}{2}}\frac{\tanh[y\sqrt{\eta\gamma}]}{\cosh[y\sqrt{\eta
\gamma}]+e^{\eta\gamma/2+\ln\left(\frac{1-\alpha}{\alpha}\right)}}\textrm{d}y\label{eq:apex}
\end{align}
Consider now the following inequalities
\begin{align*}
\tanh(z)\leq 1 \quad \text{and} \quad 
\cosh(z)\geq\frac{e^{z}}{2}
\end{align*}
After substitution and manipulation of the denominator of \eqref{eq:apex}, we obtain
\begin{align*}
\alpha\left(\zeta_{1}-\zeta_{\alpha}\right)\leq &
(1-\alpha)e^{\eta\gamma/2}\int_{-\infty}^{\infty}\frac{e^{\frac{-(y-\sqrt{\eta\gamma})^2}{2}}}{\sqrt{2\pi}}
\frac{1}{\frac{e^{y\sqrt{\eta
\gamma}}}{2}+e^{\eta\gamma/2+\ln\left(\frac{(1-\alpha)}{\alpha}\right)}}\textrm{d}y\\
=&(1-\alpha)e^{\eta\gamma/2}\int_{-\infty}^{\infty}\frac{e^{\frac{-(y-\sqrt{\eta\gamma})^2}{2}}}{\sqrt{2\pi}}
\frac{e^{-\frac{y\sqrt{\eta\gamma}}{2}-\frac{\eta\gamma}{4}-\phi_{\alpha}}}{\cosh\left(\frac
{y\sqrt{\eta\gamma}}{2}-\frac{\eta\gamma}{4}-\phi_{\alpha}\right)}\textrm{d}y,
\end{align*}
where
$\phi_{\alpha}\triangleq\frac{1}{2}\ln\left(\frac{2(1-\alpha)}{\alpha}\right)$.
We readjust terms to express the integral in a convenient form
\begin{equation*}
\alpha\left(\zeta_{1}-\zeta_{\alpha}\right) \leq
(1-\alpha)e^{-\eta\gamma/8-\phi_{\alpha}}\int_{-\infty}^{\infty}
\frac{e^{\frac{-\left(y-\frac{\sqrt{\eta\gamma}}{2}\right)^2}{2}}}{\sqrt{2\pi}}\sech\left(\frac
{y\sqrt{\eta\gamma}}{2}-\frac{\eta\gamma}{4}-\phi_{\alpha}\right)\textrm{d}y.
\end{equation*}
We now use the following transformation in the variable of
integration:
$y=2z+\frac{\sqrt{\eta\gamma}}{2}+\frac{2\phi_{\alpha}}{\sqrt{\eta\gamma}}$.
\begin{align*}
\alpha\left(\zeta_{1}-\zeta_{\alpha}\right) \leq
&2(1-\alpha)e^{\frac{7\eta\gamma}{8}-\phi_{\alpha}}\int_{-\infty}^{0}\frac{e^{\frac{-\left(2
\hat
z-2\sqrt{\eta\gamma}-\frac{2\phi_{\alpha}}{\sqrt{\eta\gamma}}\right)^2}{2}}}{\sqrt{2\pi}}\sech\left(
z\sqrt{\eta\gamma}\right)\textrm{d}\hat z\\
+&2(1-\alpha)e^{-\eta\gamma/8-\phi_{\alpha}}\int_{0}^{\infty}\frac{e^{\frac{-\left(2z+\frac{2\phi_{\alpha}}
{\sqrt{\eta\gamma}}\right)^2}{2}}}{\sqrt{2\pi}}\sech\left(
z\sqrt{\eta\gamma}\right)\textrm{d}z
\end{align*}
and the asymptotic expansion for $\sech(z)$ in the above
derivation
\begin{equation*}
\sech(z)=
2e^{-\abs{z}}\left(1+\sum_{\ell=1}^{\infty}(-1)^{\ell}e^{-2\ell\abs{z}}\right).
\end{equation*}
This yields
\begin{align*}
\alpha\left(\zeta_{1}-\zeta_{\alpha}\right)\leq
&4(1-\alpha)e^{-\eta\gamma/8-\phi_{\alpha}-\frac{2\phi^2_{\alpha}}{\eta\gamma}}
\sum_{\ell=0}^{\infty}(-1)^{\ell}\int_{0}^{\infty} \frac{e^{-2\hat
z^2+\left((1+2\ell)\sqrt{\eta\gamma}-\frac{4\phi_{\alpha}}{\sqrt{\eta\gamma}}\right)\hat
z}}{\sqrt{2\pi}}\textrm{d}z\\
&+4(1-\alpha)e^{-\eta\gamma/8-\phi_{\alpha}-\frac{2\phi^2_{\alpha}}{\eta\gamma}}
\sum_{\ell=0}^{\infty}(-1)^{\ell}\int_{0}^{\infty} \frac{e^{-2\hat
z^2-\left((1+2\ell)\sqrt{\eta\gamma}+\frac{4\phi_{\alpha}}{\sqrt{\eta\gamma}}\right)\hat
z}}{\sqrt{2\pi}}\textrm{d}z
\end{align*}
Finally, expressing the integrals in terms of the Q function
\begin{align*}
\alpha\left(\zeta_{1}-\zeta_{\alpha}\right)\leq
2(1-\alpha)e^{-\eta\gamma/8-\phi_{\alpha}-\frac{2\phi^2_{\alpha}}{\eta\gamma}}
\Biggl[&\sum_{\ell=0}^{\infty}(-1)^{\ell}e^{\frac{\left(\frac{(1+2\ell)\sqrt{\eta\gamma}}{2}
-\frac{2\phi_{\alpha}}{\sqrt{\eta\gamma}}^2\right)^2}{2}} Q\left(
\frac{\left(1+2\ell\right)\sqrt{\eta\gamma}}{2}-\frac{2\phi_{\alpha}}{\sqrt{\eta\gamma}}\right)\\
+&\sum_{\ell=0}^{\infty}(-1)^{\ell}e^{\frac{\left(\frac{(1+2\ell)\sqrt{\eta\gamma}}{2}
+\frac{2\phi_{\alpha}}{\sqrt{\eta\gamma}}\right)^2}{2}}Q\left(
\frac{(1+2\ell)\sqrt{\eta\gamma}}{2}+\frac{2\phi_{\alpha}}{\sqrt{\eta\gamma}}\right)\Biggl]
\end{align*}
and manipulating the expansion
\begin{align*}
\alpha\left(\zeta_{1}-\zeta_{\alpha}\right)\leq
&2(1-\alpha)e^{-\eta\gamma/8-\phi_{\alpha}-\frac{2\phi^2_{\alpha}}{\eta\gamma}}\Biggl[
e^{\frac{\left(\frac{\sqrt{\eta\gamma}}{2}-\frac{2\phi_{\alpha}}{\sqrt{\eta\gamma}}\right)^2}{2}}
Q\left(
\frac{\sqrt{\eta\gamma}}{2}-\frac{2\phi_{\alpha}}{\sqrt{\eta\gamma}}\right)\\
&+e^{\frac{\left(\frac{\sqrt{\eta\gamma}}{2}+\frac{2\phi_{\alpha}}{\sqrt{\eta\gamma}}\right)^2}{2}}
Q\left(
\frac{\sqrt{\eta\gamma}}{2}+\frac{2\phi_{\alpha}}{\sqrt{\eta\gamma}}\right)\\
&+\sum_{\ell=1}^{\infty}(-1)^{\ell}
e^{\frac{\left(\frac{(1+2\ell)\sqrt{\eta\gamma}}{2}-\frac{2\phi_{\alpha}}{\sqrt{\eta\gamma}}\right)^2}{2}}
Q\left(
\frac{(1+2\ell)\sqrt{\eta\gamma}}{2}-\frac{2\phi_{\alpha}}{\sqrt{\eta\gamma}}\right)\\
&+\sum_{\ell=1}^{\infty}(-1)^{\ell}
e^{\frac{\left(\frac{(1+2\ell)\sqrt{\eta\gamma}}{2}+\frac{2\phi_{\alpha}}{\sqrt{\eta\gamma}}\right)^2}{2}}
Q\left(
\frac{(1+2\ell)\sqrt{\eta\gamma}}{2}+\frac{2\phi_{\alpha}}{\sqrt{\eta\gamma}}\right)\Biggl],
\end{align*}
we obtain, after using the series expansion~\eqref{ap1},
\begin{align*}
\alpha\left(\zeta_{1}-\zeta_{\alpha}\right)\leq &
2\sqrt{\frac{(1-\alpha)\alpha}{\pi\eta\gamma}}
e^{-\eta\gamma/8}\Biggl[
\frac{1}{\left(1-\frac{4\phi_{\alpha}}{\eta\gamma}\right)}
+\frac{1}{\left(1+\frac{4\phi_{\alpha}}{\eta\gamma}\right)}-
\frac{1}{\left(3-\frac{4\phi_{\alpha}}{\eta\gamma}\right)}-
\frac{1}{\left(3+\frac{4\phi_{\alpha}}{\eta\gamma}\right)}\\
&+\dots +\mathcal{O}\left(\frac{1}{\sqrt{\eta\gamma}}\right)\Biggl],
\end{align*}
where the linear term in
$\phi_{\alpha}=\frac{1}{2}\ln\left(\frac{2(1-\alpha)}{\alpha}\right)$
is substituted in the common factor, and quadratic terms are
neglected.

Using the same result as before on the series
$2\sum_{n=0}^{\infty}\frac{(-1)^{n+1}}{2n+1}$, we obtain the
upper bound:
\begin{align*}
\alpha\left(\zeta_{1}-\zeta_{\alpha}\right)<\sqrt{\frac{\pi\alpha(1-\alpha)}
{\eta\gamma}}e^{-\eta\gamma/8}.
\end{align*}
Finally, using the upper bound given in Lemma \ref{lem1} for BPSK, the overall
MMSE can be upper-bounded by
\begin{align*}
\textrm{MMSE}_{\alpha}\leq 2\alpha e^{-\eta\gamma/2}+
\sqrt{\frac{\pi\alpha(1-\alpha)}{\eta\gamma}}e^{-\eta\gamma/8}.
\end{align*}

\section{Proof of Theorem \ref{thm4}}\label{app:beta_proof}
We analyze the function
\begin{equation}\label{eq:ap4}
 G(\eta)=(1-\eta)e^{u\eta}/\sqrt{\eta}
\end{equation}
where $u$ is a constant, which entirely describes the dependence
of $\Upsilon_\beta$ on $\eta$ for sufficiently large $\eta\gamma$. By
simple derivation of~\eqref{eq:ap4}, it is easy to observe that
the solution has critical points in the domain $(0,1]$ if and only
if $u\geq(3+\sqrt{8})/2$. These points are
\begin{eqnarray*}
\eta_m&=&\frac{u-1/2-\sqrt{u^2-3u+1/4}}{2u}\\
\eta_M&=&\frac{u-1/2+\sqrt{u^2-3u+1/4}}{2u}
\end{eqnarray*}
and lie in the domain $(0,1]$. 
In fact, note that $u^2-3u+1/4<(u-3/2)^2$, and thus it can be
verified that $0<1/2u<\eta_m<\eta_M<1-1/u<1$. By using the second
derivative of~\eqref{eq:ap4} we observe that these solutions
correspond to a local minimum and a local maximum, respectively.
Let us now study the function~\eqref{eq:ap4} to justify the range
of values of the critical system load. $G(\eta)$ is a continuous
function in $(0,1]$ that takes positive values. Since $G(\eta)$
tends to 0 as $\eta$ approaches 1, and tends to infinity as $\eta$
approaches 0, it can be concluded that the range for which
$G(\eta)$ has only one pre-image is
$(0,G(\eta_{m}))\cup(G(\eta_{M}),\infty)$. Hence, there are single
pre-images in the ranges $(0, \min \{G^{-1}(G(\eta_M))\})$ and
$(\max \{G^{-1}(G(\eta_m))\},1]$. For $G(\eta_{m})<G(.)<
G(\eta_{M})$, there are three pre-images and for 
$G(\eta_{m})$ and $G(\eta_{M})$ there are exactly two
 due to
the local minimum and maximum (See Fig. \ref{fig5b}). Then, the smallest pre-image among them lies
on $[\min \{G^{-1}(G(\eta_M))\},\eta_m]$ whereas the largest lies
on $[\eta_M, \max \{G^{-1}(G(\eta_m))\}]$. In conclusion, the
overall \emph{smallest} pre-images belong to
\begin{equation}
 \mathcal{R}_1=(0,\min \{G^{-1}(G(\eta_M))\}) \cup [\min \{G^{-1}(G(\eta_M))\},\eta_m]= (0,\eta_m],
\end{equation} 
whereas the \emph{largest} pre-images belong to 
\begin{equation*}
\mathcal{R}_2=[\eta_M, \max \{G^{-1}(G(\eta_m))\}]\cup [\max \{G^{-1}(G(\eta_m))\},1]=[\eta_M,1].
\end{equation*}
In particular,  
$\mathcal{R}_{12}\triangleq(0,\min
\{G^{-1}(G(\eta_M)\})\subset \mathcal{R}_1$ and $\mathcal{R}_{22}\triangleq(\max
\{G^{-1}(G(\eta_m))\},1]\subset\mathcal{R}_2$.
By bounding the MMSE using \eqref{eq:2.6} and
replacing $u=\gamma/8$, we obtain the desired results.

\bibliographystyle{IEEE}

\begin{thebibliography}{27}

\bibitem{Ver98}
S.~Verd\'u,
\newblock {\em Multiuser Detection}.
\newblock Cambridge University Press, UK, 1998.

\bibitem{Half96}
K.~Halford and M.~Brandt-Pearce,
\newblock ``New-user identification in a {CDMA} system,''
\newblock {\em IEEE Trans. Commun.}, vol. 46, no.\ 1, pp. 144--155, Jan. 1998.


\bibitem{Osk02}
T.~Oskiper and H.~V.~Poor,
\newblock ``Online Activity Detection in a Multiuser
Environment Using the Matrix {CUSUM} Algorithm,''
\newblock {\em IEEE Trans.\ Inform.\ Theory}, vol.\ 48, no.\ 2, pp.\ 477--493,
Feb. 2002.


\bibitem{Chen01}
B.~Chen and L.~Tong,
\newblock ``Traffic-Aided
Multiuser Detection for Random-Access {CDMA} Networks,''
\newblock {\em IEEE Trans.\ Sig.\ Process.}, vol.\ 49, no.\ 7, pp.\ 1570--1580,
July 2001.

\bibitem{Tau082}
E.~Biglieri, E.~Grossi, M.~Lops, and A.~Tauste Campo,
\newblock ``Large-system analysis of a dynamic {CDMA} system under a {M}arkovian input process,''
\newblock in {\em Proc.\ IEEE Int.\ Symp.\ Inform. Theory}, Toronto, Canada,
July 2008.


\bibitem{Poor98}
M.~L. Honig and H.~V. Poor,
\newblock ``Adaptive interference suppression in wireless communication
  systems,''
\newblock in {\em Wireless Communications: Signal Processing Perspectives}.
  1998, H.V. Poor and G.W. Wornell, Eds Englewood Cliffs, NJ: Prentice Hall.

\bibitem{Big1}
E.~Biglieri and M.~Lops,
\newblock ``Multiuser detection in a dynamic environment- {P}art {I}: User
  identification and data detection,''
\newblock {\em IEEE Trans.\ Inform.\ Theory}, vol.\ 53, no.\ 9, pp.\ 3158--3170,
  Sept. 2007.

\bibitem{Tan02}
T.~Tanaka,
\newblock ``A statistical-mechanics approach to large-system analysis of {CDMA}
  multiuser detectors,''
\newblock {\em IEEE Trans.\ Inform.\ Theory}, vol.\ 48, no.\ 11, pp.\ 2888--2910,
Nov. 2002.

\bibitem{Guo05}
D.~Guo and S.~Verd\'u,
\newblock ``Randomly spread {CDMA}: Asymptotics via statistical physics,''
\newblock {\em IEEE Trans. Inform. Theory}, vol. 51, no. 6, pp. 1983--2007,
June 2005.

\bibitem{Guo10}
D.~Guo, Y.~Wu, S.~Shamai (Shitz) and S.~Verd\'u,
\newblock ``Estimation in Gaussian noise: Properties of the minimum mean-square error,''
\newblock {\em Submitted to IEEE Trans. Inform. Theory}, http://arxiv.org/abs/1004.3332.

\bibitem{Nis92}
H.~Nishimori,
\newblock {\em Statistical Physics of Spin Glasses and Information Processing:
  An Introduction},
\newblock Oxford Univ. Press, Oxford, UK, 2001.


\bibitem{Tan08}
D.~Guo and T.~Tanaka,
\newblock ``Generic multiuser detection and statistical physics,''
\newblock in {\em Advances in Multiuser Detection}. M.\ L.\ Honig (Ed.). John
  Wiley \& Sons, 2009.

\bibitem{Mont08}
M.~M\'ezard and A.~Montanari,
\newblock {\em Information, Physics and Computation},
\newblock Oxford Univ.\ Press, 2009.


\bibitem{Kor09}
S.~B.~Korada and N.~Macris,
\newblock ``Tight bounds on the capacity of binary input random {CDMA},''
\newblock {\em Accepted for publication in IEEE Trans.\ Inform.\ Theory}, http://people.epfl.ch/satish.korada.

\bibitem{Tau07}
A.~Tauste~Campo and E.~Biglieri,
\newblock ``Large-system analysis of static multiuser detection with an unknown
  number of users,''
\newblock in {\em Proc. of IEEE Int.\ Workshop on Comp.\ Advances in Multi-Sensor
  Adapt.\ Process.\ (CAMSAP'07)}, Saint Thomas, US, Dec. 2007.

\bibitem{Tau08}
A.~Tauste~Campo and E.~Biglieri,
\newblock ``Asymptotic capacity of a static multiuser channel with an unknown
  number of users,''
\newblock in {\em Proc.\ IEEE Int.\ Symp.\ on Wir.\ Per.\ Mult.\ Comms. (WPMC'08)},
  Sept. 2008.

\bibitem{Lop06}
E.~Biglieri and M.~Lops,
\newblock ``A new approach to multiuser detection in mobile transmission
  systems,''
\newblock in {\em Proc.\ IEEE Int.\ Symp.\ on Wir.\ Per.\ Mult.\ Comms.\ (WPMC'06)},
  San Diego, CA, Sept. 2006.

\bibitem{Tse99}
D.~Tse and S.~Hanly,
\newblock ``Linear multiuser receivers: effective interference, effective
  bandwidth, and user capacity,''
\newblock {\em IEEE Trans. Inform. Theory}, vol.\ 45, no.\ 2, pp.\ 641--657,
Mar. 1999.

\bibitem{Mull03}
R.~R. Mller,
\newblock ``Channel capacity and minimum probability of error in large dual
  antenna array systems with binary modulation,''
\newblock {\em IEEE Trans. Sig. Process.}, vol.\ 51, no.\ 11, pp.\ 2821--2828,
  Nov. 2003.

\bibitem{El85}
R. S. Ellis,
\newblock {\em Entropy, Large Deviations, and Statistical Mechanics}, vol.\ 271,
\newblock A series of comprehensive studies in mathematics. Springer-Verlag,
  1985.

\bibitem{Poor88}
V.~Poor,
\newblock {\em An introduction to Signal Detection and Estimation},
\newblock Springer-Verlag, New York, 1988.

\bibitem{Tse00}
D.~Tse and S.~Verd\'u,
\newblock ``Optimum asymptotic multiuser efficiency of randomly spread {CDMA},''
\newblock {\em IEEE Trans.\ Inform.\ Theory}, vol.\ 46, no.\ 11, pp.\ 2718--2722,
Nov.\ 2000.

\bibitem{Verd02}
S.~Verd\'u and S.~Shamai (Shitz),
\newblock ``Spectral efficiency of {CDMA} with random spreading,''
\newblock {\em IEEE Trans.\ Inform.\ Theory}, vol.\ 45, no.\ 2, pp.\ 622--640,
  Mar. 1999.

\bibitem{Tal03}
M.~Talagrand,
\newblock {\em Spin Glasses: A Challenge for Mathematicians},
\newblock Springer, 2003.

\bibitem{Mont06}
A.~Montanari and D.~Tse,
\newblock ``Analysis of belief propagation for non-linear problems: The example
  of {CDMA} (or: How to prove {T}anaka's formula),''
\newblock in {\em Proc.\ IEEE Inform.\ Theory Workshop}, Punta del Este, Uruguay,
  Mar. 2006.

\bibitem{Guo07}
D.~Guo and C.-C.~Wang,
\newblock ``Random sparse linear systems observed via arbitrary channels: A
  decoupling principle,''
\newblock in {\em Proc.\ IEEE Int.\ Symp.\ Inform. Theory}, Nice, France, June
  2007.
  
\bibitem{Guo09}
D.~Guo, D. Baron and S. Shamai (Shitz),
\newblock ``A single-letter characterization of optimal noisy compressed sensing,''
\newblock in {\em Proc.\ Allerton Conf.\ Commun.\, Control, and Computing}, Oct
2009.


\bibitem{Loz06}
A.~Lozano, A.~M.\ Tulino, and S.~Verd\'u,
\newblock ``Optimum power allocation for parallel gaussian channels with
  arbitrary input distributions,''
\newblock {\em IEEE Trans.\ Inform.\ Theory}, vol.\ 52, no.\ 7, pp.\ 3033--3051, July
  2006.

\bibitem{TauGui10}
A. Tauste Campo and A. Guill\'en i F\`abregas,
\newblock ``Large System Analysis of Iterative Multiuser Joint Decoding with an Uncertain Number of Users,''
\newblock in {\em Proc. IEEE Int. Symp. on Inform. Theory}, Austin, Texas, June
  2010.

\end{thebibliography}

\end{document}